\documentclass[conference, 10pt]{IEEEtran}
\usepackage{float}
\usepackage{algorithm, algpseudocode}
\usepackage{graphicx}
\usepackage{caption}
\usepackage{color}
\usepackage{subcaption}
\usepackage{tabularx}
\usepackage{amsmath,amsthm,amssymb,amstext}
\usepackage{mathtools}
\usepackage{comment}
\usepackage{multirow}
\usepackage{cite}
\newtheorem{theorem}{Theorem}
\newtheorem{definition}{Definition}
\newtheorem{lemma}{Lemma}

\algnewcommand\INPUT{\item[\textbf{Input:}]}%
\algnewcommand\OUTPUT{\item[\textbf{Output:}]}%

\DeclarePairedDelimiter\ceil{\lceil}{\rceil}

\begin{document}
\title{Optimal User-Cell Association for 360 Video Streaming over Dense Wireless Networks}
\author{
\IEEEauthorblockN{Po-Han Huang and Konstantinos Psounis}
\IEEEauthorblockA{Ming Hsieh Department of Electrical Engineering\\University of  Southern California, Los Angeles, CA, USA, 90089\\
Email: \{pohanh, kpsounis\}@usc.edu}
}

\maketitle

\begin{abstract}
Delivering 360 degree video streaming for virtual and augmented reality presents many technical challenges especially in bandwidth starved
wireless environments.
Recently, a so-called two-tier approach has been proposed which delivers a basic-tier chunk and select enhancement-tier chunks 
to improve user experience while reducing network resources consumption.
The video chunks are to be transmitted via unicast or multicast over an ultra-dense small cell infrastructure with enough bandwidth 
where small cells store video chunks in local caches. 
In this setup, user-cell association algorithms play a central role to efficiently deliver video since users may only download video chunks
from the cell they are associated with.
Motivated by this, we jointly formulate the problem of user-cell association and video chunk multicasting/unicasting
as a mixed integer linear programming, prove its NP-hardness, and study the optimal solution via the Branch-and-Bound method.
We then propose two polynomial-time, approximation algorithms and show via extensive simulations that they are 
near-optimal in practice and improve user experience by 30\% compared to baseline user-cell association schemes.
\end{abstract}

\section{Introduction}
\label{Intro}
Delivering 360 video streaming for virtual and augmented reality (VR/AR) is the next big challenge in wireless networks due to the associated high-bandwidth demand and 
the dynamic nature of such applications \cite{Trestian:Survey,Argyriou:Network2017}. 
TTo make this come true, several research directions have been investigated, the most promising ones being the reduction of the amount of data for delivery and the increase of the available bandwidth via new wireless technologies. 
To reduce the amount of data for delivery, two-tier 360 video systems and tile-based 360 video streaming have been recently proposed in \cite{Qian:AllThingsCellular2016, Chen:GLOBECOM2017, Chakareski:VR/AR2017, Ghosh:arXiv2017, Duanmu:VR/AR2017, Corbillon:MM2017, Sun:MMSys2018}. 
The main idea of these technologies is to divide the whole video into a basic-tier chunk and multiple enhancement-tier chunks (or tiles), and deliver the basic-tier chunk and a portion of enhancement-tier chunks to the users based on their specific requirements and the placement of the enhancement-tier chunks. The basic-tier chunk is used to ensure video availability from any angle while the enhancement-tier chunk is used to improve user experience. 

To increase wireless bandwidth the industry is envisioning 5G systems with dense base station deployments. Such small cell architectures may support wireless virtual reality \cite{Argyriou:Network2017, Huang:WONS2017, Huang:COMCOM2019}
as their increased density, which can be up to 75-200 small cells per square km, see, for example, reports from the 5G America/Small Cell Forum \cite{5GAmerica,Nokia}, can provide a significantly higher system capacity by frequency reuse. A major challenge in such ultra dense deployments is user-cell association due to high traffic variability \cite{Liu:Survey, Haw:Access2019, Fooladivanda:TWC2013, Wang:Access2016, Sapountzis:INFOCOM2018, Ao:Mobihoc2016}, a problem
exacerbated by VR/AR applications. Traditional signal-to-interference-noise-based (SINR-based) user-cell association schemes may not be able to support this kind of applications and meet the quality of experience expected by users.

\begin{figure}
\vspace{0.2cm}
\centering
\includegraphics[scale=0.3]{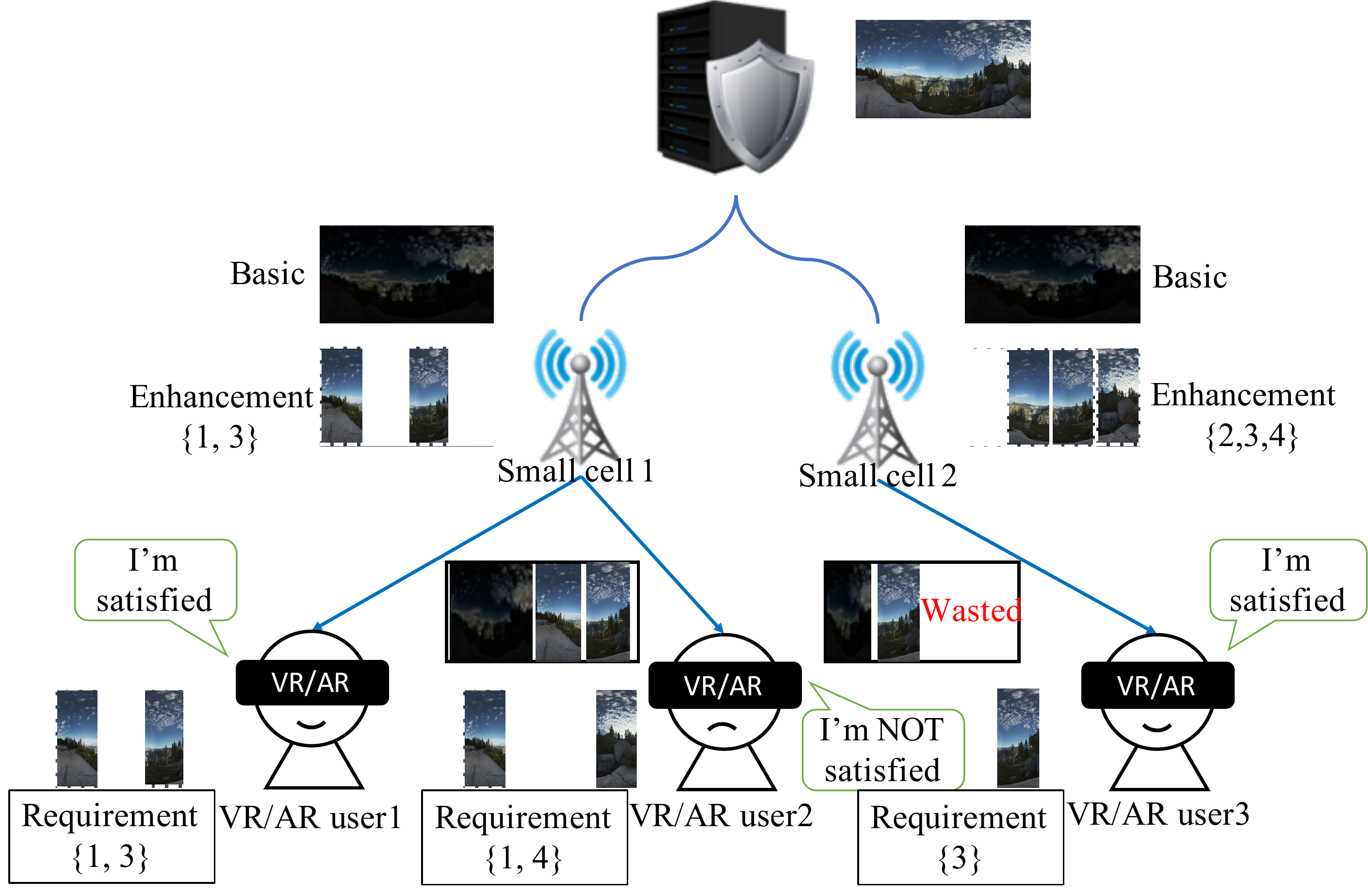}
\caption{Simple scenario for two-tier 360 video delivery over wireless networks.}
\label{fig:Scenario}
\centering
\end{figure}

Consider Fig. \ref{fig:Scenario} as a simplified example. There are two small cells with some cache capacity which are connected to a content server, and, there are three VR/AR users 
requiring some video. 
Since the small cells have limited capacity, assume small cell 1 has stored in its cache the basic view (basic-tier chunk) and enhanced views (enhanced-tier chunks) 1 and 3, and small cell 2 
has the basic view and enhanced views 2, 3, and 4. 
The three users may require different enhanced views depending on their interests, direction of walking/gaze, etc. For example, say user 1 wants views 1 and 3, user 2 wants views 1 and 4, 
and, user 3 wants view 3 only. 
It is easy to see that SINR-based user-cell association may result in suboptimal operation as it is agnostic to the cache content of the various small cells. In our example, users 1 and 2 obtain 
service from small cell 1, user 3 obtains service from small cell 2, and while users 1 and 3 can receive their required views user 2 can only receive one of his/her required views, even though the system
has enough capacity to satisfy all users. 

There are real world VR/AR applications which are expected to have this challenge at a much larger scale. 
For example, the National Basketball Association (NBA) would like to have VR support in each game \cite{NBA:VR}, where typical stadiums have tens of thousands of seats over tens of thousands of 
square meters, e.g., the Staples Center has 21,000 seats in 88,257.9 $m^2$. 
As another example, augmented vehicular reality (AVR) \cite{AVR} aims to make autonomous driving safer by extending vehicle's visual horizon via sharing visual information with other nearby vehicles. 
These examples require to understand how to associate every user/vehicle with appropriate cells/content proxies given communication bandwidth limitations.

Motivated by the above, in this paper we investigate user-cell association schemes for 360 video streaming applications, where basic-tier and enhanced-tier video chunks
are transmitted via broadcast, multicast or unicast transmissions from small cells to users.
Since users can only receive data from the cell they are associated with and different cells will cache different content, it is evident that a joint optimization of 
user-cell association, video chunk placement, and selection of the enhanced views to be transmitted to users under a bandwidth constraint, is required to maximize
the user experience. Also, since 360 video, when available, will be the main bandwidth consumer, it is reasonable to consider it in the user-cell association decisions.
We formulate the associated joint optimization problem as a mixed integer linear programming and prove its NP-hardness by reducing it to a binary multiple knapsack problem (MKP).
We then propose three algorithms to solve the problem: one based on the Branch-and-Bound method (BB) which yields the optimal solution, and two polynomial-time 
approximation algorithms, a submodular-based greedy algorithm which we refer to as \textbf{E}fficient \textbf{L}ayered \textbf{V}ideo delivery \textbf{A}lgorithm (ELVA), 
and an MKP-based greedy algorithm which we refer to as \textbf{E}fficient \textbf{V}ideo delivery \textbf{A}lgorithm (EVA).
We show that ELVA is a $\rho(1-\frac{1}{e})$-approximation algorithm, where $\rho$ is given by the network topology, and EVA is a $\frac{1}{2}$-approximation algorithm. 
Simulation results from both small-scale settings and large-scale settings show that both algorithms outperform baseline user-cell association schemes in all scenarios. 
Noteworthy, ELVA has a near-optimal performance and can increase the quality of the user experience by 30\% compared to SINR-based user-cell association.

The rest of this paper is organized as follows. 
Section \ref{ReWork} summarizes the related work, Section \ref{SysModel} presents the system architecture and the main assumptions of our model,
Section \ref{Problem} formulates the problem at hand, Section \ref{Algorithm} presents the three aforementioned algorithms and formally studies their performance,
and Section \ref{Simulation} evaluates the performance of the proposed algorithms using simulations under a variety of realistic scenarios.
Last, Section \ref{Conclusion} concludes the paper.

\section{Related Work} 
\label{ReWork}

\noindent
{\bf VR/AR Content Delivery over Wireless Networks:}
The efficient delivery of mobile video content is becoming one of the highest priorities for the emerging ultra-dense small-cell deployments \cite{Argyriou:Network2017}.
Especially for VR/AR applications, a cellular-friendly streaming scheme was studied in \cite{Qian:AllThingsCellular2016}, where the authors demonstrated that delivering only the visible portion of views can reduce bandwidth consumption without significant degradation of the user experience. 
Based on the above idea, the authors in \cite{Chen:GLOBECOM2017} proposed a resource management mechanism to further improve user experience.
Then, the authors in \cite{Chakareski:VR/AR2017} investigated how to place caches in the cells such that performance (measured in rewards earned by the service provider) is maximized. 

\noindent
{\bf Multi-Tier Video Streaming:}
There are two approaches to deal with head movement, a key challenge in 360-degree video delivery. The first is to predict users' behavior and thus upcoming head movement, see, for example, \cite{Qian:AllThingsCellular2016,Qian:Mobicom2018,Park:Networking2019}, such that the system can proactively download future video angles. The second is to use a multi-tier 360 video system, see, for example, 
\cite{Ghosh:arXiv2017,Duanmu:VR/AR2017,Corbillon:MM2017,Sun:MMSys2018,Li:IPSN2019}, such that a low resolution basic-tier 360-degree video is always available for all angles, 
giving to the system more time to download enhanced video chunks for specific view angles. The two approaches are complementary. 
Our work reduces network resource consumption and thus applies to both approaches. Since the later approach tends to have higher bandwidth requirements, we use it as our use-case application. 

\noindent
{\bf User Association:}
Many prior works suggested that instead of SINR-based user-cell association (to be referred to as user association for brevity henceforth), application-aware or resource-aware 
user association scheme could improve user experience \cite{Liu:Survey, Haw:Access2019}.
For example, joint resource allocation and user association was investigated in \cite{Fooladivanda:TWC2013}  and a significant gain in terms of system throughput has been shown.
As another example, cache-aware user association was studied in \cite{Wang:Access2016} with the goal to minimize delay. 
Moreover, joint optimization of user association and dynamic TDD was discussed in \cite{Sapountzis:INFOCOM2018} where the authors showed that the system could achieve 
higher throughput for both downlink and uplink traffic by optimizing the TDD schedule and the user association at the same time.  

\noindent
{\bf Hybrid Multicast/Unicast in Wireless Networks:}
Evolved Multimedia Broadcast Multicast Services (eMBMS) was standardized in LTE network \cite{3GPP:23.246} using a fixed number of resources. 
Recently, many studies have proposed unicast services and dynamic resource allocation for eMBMS to enhance resource utilization and system performance.
For example, \cite{Monserrat:TB2012} studied dynamic eMBMS for speeding up file delivery, \cite{Chen:INFOCOM2015} proposed an efficient user grouping mechanism 
in the presence of hybrid multicast and unicast eMBMS services to achieve better system throughput, and \cite{Hou:ISM2016} investigated a hybrid multicast and unicast service
for a VR application where the authors proposed a novel hyper-cast approach to reduce total bitrate and save wireless network bandwidth.

Despite the large body of work in the aforementioned distinct research areas, the performance of the overall multi-tier video delivery system jointly depends on user-cell association, 
as users can only receive data from their cell, cache and content placement, as cells cannot realistically store the basic and all enhanced views of all videos of interest, bandwidth 
resource allocation, as wireless bandwidth and the associated resource blocks are limited, and the selection of enhanced views to be transmitted to users via multicasting or unicasting. 
This joint optimization problem is the topic of this work.

\section{System Architecture} 
\label{SysModel}

\subsection{Caching Model} 
\label{Caching}
Motivated by practical considerations, we assume every small cell within the neighborhood of some view, e.g. a store front, has a copy of the basic view. 
This provides fault tolerance against system/cell failures without much cost as the size of the basic view is typically small compared to the enhanced view (e.g. 0.57/0.42 Mb/s for raw/compressed 
1080p resolution video and 2.3 Mb/s for 2K resolution video \cite{Petrangeli:MM2017}). What is more, since VR/AR users are usually in the same immersive environment, e.g., playing the same 
VR game, watching the same NBA game, etc., the basic view is indeed common to all whereas the enhanced views are usually personalized due to user-specific interests and/or the field of 
view in the users' head-mounted display (HMD). 
As a consequence, a user can always get the content of the basic view from the small cell one associates with.

Let $K$ be the maximum size of the cache in a small cell, measured by the number of enhanced views that may fit in a small cell.
For example, in Fig. \ref{fig:Scenario}, $K=2$ for small cell 1, and $K=3$ for small cell 2. 
For the enhanced views caching placement, we can apply any schemes proposed by prior work, e.g., \cite{FemtoCaching,PTAM,Ao:TMC2018,Maniotis:ArXiv2019,Mahzari:MM2018,Lin:ICC2020}. 
Under a given placement of enhanced views to caches and given bandwidth constraints, our goal is to jointly optimize user-cell association and the selection of enhanced views to be transmitted 
(via multicast or unicast) to users, such that user experience is maximized (see later for a formal definition).
Note that the caching placement and the user-cell association are typically on different time scales. Cellular service providers will infrequently reallocate enhanced views among 5G small cell caches due to the associated cost and latency to transfer these data among small cells or download them from the cloud, 
whereas user-cell association is expected to change frequently in the context of 5G ultra dense networks, see, for example \cite{Haw:Access2019,Ibrahim:TCOM2016,Kwak:GLOBECOM2017}.
For this reason, we jointly optimize the user cell association and enhanced views transmission scheduling given a caching placement, rather than also jointly optimizing the later.

\subsection{Multiple Description Video Streaming} 
\label{MDC}
We adopt the model of multiple description coding (MDC) in \cite{Zhou:TB2009} for enhanced views, which is widely used for mobile clients. 
With MDC, a small cell receives an enhanced view $k$ whose highest resolution version is of size $E_k$, as well as lower resolutions of the same enhanced view.
To simplify notation, instead of indexing the different resolution versions of $k$, we use a number between 0 and 1 to indicate the size of a lower
resolution version of $k$ as a fraction of $E_k$, e.g. 0.5$E_k$.
Clearly, the user experience is proportional to the resolution of the enhanced view that a user receives \cite{Singh:VCIP2000,Ghareeb:CNSR2010,Kazemi:MS2014}. 

\subsection{Unicast versus Multicast Transmissions}
Whether one may unicast or multicast enhanced views depends on the application.
To see this, note that for VR/AR applications, it is common that the system fuses user-specific information/metadata into an enhanced view. 
For example, if a user is playing a game in a zombie mode which all the characters in the game are zombies, and another user is playing in an elf mode which all the characters in the game are elves, both users are seeing the same game object but with different styles. 
As another example, students may read the same material with a customized presentation style in an immersive classroom, see \cite{Hou:ISM2016}. 
Since the image transfer/fusion is a computation-intensive task which VR devices (e.g., HMDs) do not support, a small cell with edge computing capabilities should take care of this. 
Hence, it makes more sense for a small cell to unicast the fused enhanced views to each user. 
That said, there is always a case that several users share the same metadata, or that there are no metadata to be fused with an enhanced view, in which case it makes sense to 
multicast the enhanced view. We start the analysis by considering unicast mode (Section \ref{sec:ilp}) and extend it to cover multicast (Section \ref{Multicasting}).

\subsection{Wireless Model}
The data rate from small cell $j$ to mobile user $i$ is defined as follows,
\begin{equation} \label{rate}
R_{i,j} = W\log_{2}{\left (1+\frac{P_{t_j}g_{i,j}(d_{i,j})}{\mathcal{N}+\sum_{n:n\neq j}{P_{t_n}g_{i,n}(d_{i,n})}}\right )},
\end{equation}
where $W$ is the size of the operational bandwidth, $P_{t_j}$ is the transmission power of small cell $j$, $g_{i,j}(\cdot)$ is the channel gain, which is a function of the Euclidean distance, $d_{i,j}$, between mobile user $i$ and small cell $j$, and $\mathcal{N}$ is the noise power.

In the context of cellular networks, data are transmitted via resource blocks (RBs) whose duration of transmission, $\tau$, is fixed. Thus, depending on the data rate, a different number of bits, $\tau R_{i,j}$, 
gets transmitted per resource block. Then, if $B$ is the size of a basic view measured in bits, it follows that the required number of RBs to deliver a basic view to user $i$ from small cell 
$j$, $N^b_{i,j}$, equals
$N^b_{i,j} = \ceil{\frac{B}{\tau R_{i,j}}}= \ceil{\frac{B}{R_{i,j}}}$,
assuming $\tau=1$ to simplify the notation and without loss of generality. 
Similarly, the required number of RBs to deliver enhanced view $k$ at the highest resolution to user $i$ from small cell $j$, $N^e_{i,j,k}$, equals
$N^e_{i,j,k} = \ceil{\frac{E_k}{R_{i,j}}}$.

As already discussed, depending on the application, enhanced views may be unicast to a specific user or multicast to a group of users.
In what follows we first consider the case where the system broadcasts the basic view to every user and unicasts user-specific enhanced views to different users 
and then extend the model and algorithms to account for a system which multicasts common enhanced views when they are requested by several users on the same cell.
\section{Problem Formulation} 
\label{Problem}
We wish to study the problem of optimal user association and resource allocation for two-tier 360 video delivery.
With this in mind, we formulate the problem via mixed integer linear programming and prove its NP-hardness,
and establish that the problem can be separated into smaller subproblems.


\subsection{Mixed Integer Linear Programming} 
\label{sec:ilp}
Let $\mathcal{M} = \{1,...,M\}$ be the set of mobile users, $\mathcal{S} = \{1,...,S\}$ be the set of small cells, 
and $\mathcal{E} = \{1,...,E\}$ be the set of enhanced views corresponding to a 360 degree basic view (e.g. for 45 degree enhanced views with no overlap, $E=8$). 
Also, let $w_{i,j,k} = \{0,1\}$ be an indicator function showing if the desired enhanced view $k$ of user $i$ is in small cell $j$, 
where for every small cell $j$ with cache size $K$, it must be that 
\begin{equation}
\sum_{\forall k \in \mathcal{E}} \left(1-\prod_{\forall i \in \mathcal{M}} (1-w_{i,j,k})\right) \leq K \nonumber
\end{equation}
to satisfy the cache size constraint.
Further, let $N_j$ be the total number of 
RBs available to small cell $j$ at each timeframe.
Last, let $x_{i,j} = \{0,1\}$ indicate if user $i$ is associated with small cell $j$, and $y_{i,j,k} \in [0,1]$ indicate the resolution in which user $i$ receives enhanced view $k$ from small cell ${j}$ assuming that the 
user is associated with this base station. Note that $y_{i,j,k}$ is by definition 0 if user $i$ is not associated with small cell $j$. If the user is associated with the small cell, a 0 value indicates no reception of the enhanced 
view while a 1 value indicates reception at the highest possible resolution. 
TABLE \ref{tab:notation} summarizes the notation.

Upon delivery of a desired enhanced view at resolution $y_{i,j,k}$, the user receives a \textit{reward} proportional to the resolution.
Our objective is to maximize the total reward from delivering enhanced views at select resolution levels to users.
With all the above in mind, we use mixed integer linear programming and formulate the problem as follows:

\begin{align}
\max_{x_{i,j},y_{i,j,k}}&~~\sum_{\forall i \in \mathcal{M}}{\sum_{\forall j \in \mathcal{S}}{\sum_{\forall k \in \mathcal{E}}{y_{i,j,k}w_{i,j,k}}}} \nonumber\\
\text{subject to:} &~~\sum_{\forall j \in \mathcal{S}}{x_{i,j}} = 1, \;\; \forall i \in \mathcal{M}, \nonumber\\
&~~y_{i,j,k} \leq x_{i,j}w_{i,j,k},  \;\; \forall i \in \mathcal{M}, \forall j \in \mathcal{S}, \forall k \in \mathcal{E}, \nonumber\\
&~~\max_{\forall i \in \mathcal{M}}\left\{x_{i,j}N^b_{i,j}\right\}+\sum_{\forall i \in \mathcal{M}}\sum_{\forall k \in \mathcal{E}}{y_{i,j,k}N^e_{i,j,k}} \leq N_j, \nonumber\\
&~~~~~~~~~~~~~~~~~~~~~~~~~~~~~~~~~~~~~~~~~~~~~~~~~~\forall j \in \mathcal{S}, \nonumber\\
&~~x_{i,j} = \{0,1\}, \forall i \in \mathcal{M}, \forall j \in \mathcal{S}, \nonumber\\
&~~y_{i,j,k} \in [0,1], \;\; \forall i \in \mathcal{M}, \forall j \in \mathcal{S}, \forall k \in \mathcal{E}. \label{p1}
\end{align}

The optimization is over $x_{i,j}$, i.e. user-cell association, and $y_{i,j,k}$, i.e. the selection and resolution of the enhanced views to be transmitted to users, given which enhanced views 
each small cell has on its cache, and the available RBs (wireless bandwidth) for each cell. 
The first constraint in the above formulation is used to indicate that a user can only associate with one small cell.
The second constraint is to assure that the portion of an enhanced view can only be delivered when the user decides to associate with the small cell and the small cell has the enhanced view.
The third constraint is to assure that the total number of required RBs to transmit the basic views (first term in the constraint where broadcasting is used) and the enhanced views
at the select resolution levels (second term in the constraint where unicast is used) can not exceed the total number of RBs available to the small cell.

\begin{table}[t]
\renewcommand{\arraystretch}{1.2}
\caption{\textsc{Problem Formulation Notation}}
\vspace{-0.2cm}
\label{tab:notation}
\centering
\begin{tabularx}{\columnwidth}{|p{6.95cm}|l|}
\hline
\textbf{Description}&\textbf{Notation}\\
\hline
\hline
Decision variable for establishing connection between user ${i}$ and small cell ${j}$ &${x_{i,j}}$\\
\hline
Decision variable for determining the resolution in which user $i$ receives enhanced view $k$ from small cell ${j}$ & ${y_{i,j,k}}$\\
\hline
A set of VR users & ${\mathcal{M}}$\\
\hline
A set of small cells & ${\mathcal{S}}$\\
\hline
A set of enhanced views & ${\mathcal{E}}$\\
\hline
Indication if the desired enhanced view $k$ of user $i$ is in small cell $j$ & ${w_{i,j,k}}$\\
\hline
The maximum size of the cache in a small cell & $K$\\
\hline
Number of RBs for user ${i}$ to get the basic view from small cell $j$ &${N^b_{i,j}}$\\
\hline
Number of RBs for user ${i}$ to get enhanced view $k$ at highest resolution from small cell $j$ & ${N^e_{i,j,k}}$\\
\hline
Total number of RBs for small cell $j$ of each timeframe & ${N_j}$\\
\hline
\end{tabularx}
\end{table}

\subsection{Problem Analysis}
\subsubsection{Complexity Analysis}
We prove that the problem is NP-hard by reducing it to a 0-1 Multiple Knapsack Problem. 
\begin{definition}
(0-1 Multiple Knapsack Problem \cite{MKP}) Given a set of $n$ items and a set of $m$ knapsacks, let $f_v(j)$ be the profit of item $j$, $f_w(j)$ be the weight of item $j$, and $C_i$ be the capacity of knapsack $i$. The problem is to select $m$ disjoint subsets of items so that the total profit of the selected items is maximized, and each subset can be assigned to a corresponding knapsack whose capacity is no less than the total weight of the items in the subset. Formally,
\end{definition}
\vspace{-0.5cm}
\begin{align}
\max_{u_{i,j}} &~~\sum_{\forall i}\sum_{\forall j}{f_v(j)u_{i,j}} \nonumber \\
\text{subject to:} &~~\sum_{\forall i}\sum_{\forall j}{f_w(j)u_{i,j}} \leq C_i, \forall i \in \{1,...,m\} \nonumber \\
&~~\sum_{j}u_{i,j} \leq 1, \forall j \in \{1,...,n\} \nonumber \\
&~~u_{i,j} = \{0,1\}, \forall i, \forall j. \label{mkp}
\end{align}

\begin{lemma} \label{NP-hard}
Problem (\ref{p1}) is NP-hard.
\end{lemma}
\begin{proof}
We first show that a simplified version of Problem (\ref{p1}) is NP-hard and then argue that Problem (\ref{p1}) is also NP-hard.
Consider that $N^b_{i,j} = N^b, \forall i \in \mathcal{M}, \forall j \in \mathcal{S}$, thus obviously $\max\{N^b_{i,j}\} = N^b$.
In addition, consider that $\sum_{\forall k \in \mathcal{E}}{y_{i,j,k}N^e_{i,j,k}} = N^e_{i,j}, \forall i \in \mathcal{M}, \forall j \in \mathcal{S}$, and $\sum_{\forall k \in \mathcal{E}}{y_{i,j,k}w_{i,j,k}} = w_{i,j}, \forall i \in \mathcal{M}, \forall j \in \mathcal{S}$ with $y_{i,j,k} = 1$ if $x_{i,j} = 1$. 
We can transform our problem to the standard form of MKP by the following steps. 
a) The first constraint in Problem (\ref{p1}) is the same as the second constraint in Problem (\ref{mkp}) when $x_{i,j} = u_{i,j}$. 
b) The third constraint in Problem (\ref{p1}) can be simplified as: $\sum_{\forall i \in \mathcal{M}}\sum_{\forall j \in \mathcal{S}}{N^e_{i,j}x_{i,j}} \leq N_j - N^b$, which is the same as the first constraint in Problem (\ref{mkp}) with $x_{i,j} = u_{i,j}$, $C_i = N_j - N^b$, and $f_w(j) = N^e_{i,j}$. 
c) The objective function in Problem (\ref{p1}) can be simplified as: $\sum_{\forall i \in \mathcal{M}}\sum_{\forall j \in \mathcal{S}}{w_{i,j}x_{i,j}}$, which is the same as the objective function in Problem (\ref{mkp}) with $x_{i,j} = u_{i,j}$ and $f_v(j) = w_{i,j}$. From this, it is evident that the MKP is a special case of Problem (\ref{p1}). It is obvious that Problem (\ref{p1}) is more complicated than MKP because $N^b_{i,j}$ can be any positive number as well as there is one more decision variable $y_{i,j,k}$ involved in the problem. Since MKP is known to be NP-complete, Problem (\ref{p1}) is NP-hard.
\end{proof}
Since Problem (\ref{p1}) is NP-hard, there is no polynomial-time algorithm to solve this problem optimally. 

\subsubsection{Separability} 
\label{Separability}
When we fix $x_{i,j}$, the problem for each small cell $j$ becomes as follows:
\begin{align}
{U_j{(\mathcal{M}_j})}&~~ = \max_{y_{i,j,k}}{\sum_{\forall i \in \mathcal{M}_j}{\sum_{\forall k \in \mathcal{E}}{y_{i,j,k}w_{i,j,k}}}} \nonumber\\
\text{subject to:} &~~y_{i,j,k} \leq w_{i,j,k},  \;\; \forall i \in \mathcal{M}_j, \forall k \in \mathcal{E}, \nonumber\\
&~~\sum_{\forall i \in \mathcal{M}_j}\sum_{\forall k \in \mathcal{E}}{y_{i,j,k}N^e_{i,j,k}} \leq \overline{N}_j, \;\; \nonumber\\
&~~y_{i,j,k} \in [0,1], \;\; \forall i \in \mathcal{M}_j, \forall k \in \mathcal{E}, \label{p2}
\end{align}
where $\mathcal{M}_j = \{i \in \mathcal{M} ~ | ~ x_{i,j} = 1\}$ and $\overline{N}_j = N_j - \max_{\forall i \in \mathcal{M}_j}\left\{N^b_{i,j}\right\}$ 
(note that $x_{i,j}=1$ for all $i \in \mathcal{M}_j$). 
Therefore, Problem (\ref{p1}) is separable for each small cell $j$.
Note that Problem (\ref{p2}) is a convex optimization problem and is in a standard form of fractional Knapsack problem. Therefore, it can be easily solved by a standard greedy algorithm which 
picks the element with the highest value of $\frac{w_{i,j,k}}{N^e_{i,j,k}}$ iteratively until the second constraint becomes an equality. 

\subsubsection{Enhanced Views Multicasting} 
\label{Multicasting}
To model the case where some of the enhanced views can be multicast and some cannot, we change the third constraint in Problem (\ref{p1}) as follows:
\begin{align}
&\max_{\forall i \in \mathcal{M}}\left\{x_{i,j}N^b_{i,j}\right\} \nonumber \\
&~~+\sum_{\forall k \in \mathcal{E}}\left\{\max_{\forall i \in \mathcal{M}_k}\left\{y_{i,j,k}N^e_{i,j,k}\right\}+\sum_{\forall i \in \mathcal{M} \setminus \mathcal{M}_k}{y_{i,j,k}N^e_{i,j,k}}\right\} \nonumber\\
&~~\leq N_j, \forall j \in \mathcal{S},\nonumber
\end{align}
where $\mathcal{M}_k$ is the set of mobile users sharing enhance view $k$. 
The first term of the new constraint is the number of RBs required to broadcast the basic view, the second term is the number of RBs required to multicast enhanced view $k$ to a specific group of users, and the third term is the number of RBs required to unicast enhanced view $k$ to the rest of the users that wish to receive it. 
Clearly, multicasting some of the enhanced views can benefit the system because of the reduction of RBs usage, which we show via simulations later.
\section{Proposed Algorithms} 
\label{Algorithm}
We start with the optimal solution and then propose polynomial time approximation algorithms.
 
\subsection{Optimal Solution - the Branch-and-Bound Method}
\subsubsection{Problem Transformation}
Recall that the original problem can be decomposed into many small problems (Problem (\ref{p2})) when $x_{i,j}$ is fixed, and each small problem is a convex optimization problem (see Sec. \ref{Separability}).
Therefore, the remaining task is to determine the optimal value of $x_{i,j}$, i.e., the optimal small cell association for every user. 
We have the following binary integer programming:
\begin{align}
\max_{x_{i,j}}&~~\sum_{\forall j \in \mathcal{S}}{U_j(\mathcal{M}_j)} \nonumber\\
\text{subject to:} &~~\sum_{\forall j \in \mathcal{S}}x_{i,j} = 1, \forall i \in \mathcal{M}  \nonumber\\
&~~\mathcal{M}_j = \{i \in \mathcal{M} ~ | ~ x_{i,j} = 1\}  \nonumber\\
&~~x_{i,j} = \{0,1\}, \forall i \in \mathcal{M}, \forall j \in \mathcal{S}. \label{bb}
\end{align}

\subsubsection{Algorithm}
Since Problem (\ref{bb}) is a binary integer programming, we can use the Branch and Bound method \cite{BB} to find the optimal solution.
We define two new variables for this algorithm, $\mathcal{M}^{C} = \mathcal{M} - \cup_{\forall j \in \mathcal{S}} {\mathcal{M}_j}$ and $\mathcal{S}_i^{C} = \{j \in \mathcal{S} ~ | ~ j \text{ can be selected by } i\}$.
The algorithm proceeds as follows: 
\begin{itemize}
\item [1.] (Initialization) Set $\mathcal{M}_j = \emptyset, \forall j \in \mathcal{S}$, $\mathcal{M}^*_j = \emptyset, \forall j \in \mathcal{S}$, $\mathcal{S}_i^{C} = \mathcal{S}, \forall i \in \mathcal{M}$, and $\mathcal{V}^* = - \infty$.
\item [2.] (Branching)  Evaluate the potential 
\begin{equation}
\mathcal{P} = \sum_{\forall i \in \mathcal{M}^{C}}\left \{ \max_{\forall j \in \mathcal{S}^{C}} \left \{\sum_{\forall k \in \mathcal{E}}w_{i,j,k} \right \} \right \} \nonumber
\end{equation}
of this node. If $\mathcal{P} > \mathcal{V}^*-\sum_{\forall j \in \mathcal{S}}{U_j(\mathcal{M}_j)}$, set $x_{i^*,j^*} = 1$ for which $\{i^*,j^*\} = \arg\max_{i \in \mathcal{M}^{C}, j \in \mathcal{S}_i^{C}} \{\sum_{\forall k \in \mathcal{E}}w_{i,j,k}\}$ and $\mathcal{M}_{j^*} \leftarrow \{i^*\}$. Otherwise, go to Step 4. If the first constraint in Problem (\ref{bb}) is satisfied, i.e., $\sum_{j \in \mathcal{S}}x_{i,j} = 1, \forall i \in \mathcal{M}$, go to Step 3. Otherwise, repeat Step 2.
\item [3.] (Updating) Solve Problem (\ref{p2}) based on $\mathcal{M}_j$. If $\sum_{\forall j \in \mathcal{S}}{U_j(\mathcal{M}_j)} \geq \mathcal{V}^*$, set $\mathcal{V}^* = \sum_{\forall j \in \mathcal{S}}{U_j(\mathcal{M}_j)}$ and set $\mathcal{M}^*_j = \mathcal{M}_j$. Otherwise, go to Step 4.
\item [4.] (Backtracking) Choose the lastly joined user $\bar{i}$. Remove $\bar{i}$ from $\mathcal{M}_j$, remove $j$ from $\mathcal{S}_{\bar{i}}^{C}$, set $\mathcal{S}_i^{C} = \mathcal{S}, \forall i \in \mathcal{M}, \text{ where } i \neq \bar{i}$, and set $x_{\bar{i},j} = 0, \forall j \in \mathcal{S}$. If there is no lastly joined user, stop. Otherwise, go to Step 2.
\end{itemize}

The idea of this algorithm is efficiently going through all the user-association combinations by evaluating the potential $\mathcal{P}$ of the current state ($x_{i,j}$). 
If there exists a branch having potential to increase the value, the Branch-and-Bound method should proceed going through the user-association combinations based on the current state. 
Otherwise, if the result of the potential $\mathcal{P}$ is bounded by the current maximum value $V^*$, it is unnecessary to evaluate the rest of combinations based on the current state. 
In this case, the algorithm goes back to the previous state and tries to branch other possible combinations which have not been visited yet. 

\subsubsection{Performance Analysis}
We analyze the time complexity of the Branch-and-Bound method in the following lemma. 
\begin{lemma}
The time complexity of the Branch-and-Bound method is $O(|\mathcal{M}|^{|\mathcal{S}|}|\mathcal{S}||\mathcal{M}|\log(|\mathcal{M}|))$.
\end{lemma}
\begin{proof}
First of all, since Problem (\ref{p2}) is a fractional Knapsack problem, the time complexity to solve every subproblem (i.e., Problem (\ref{p2})) is $O(|\mathcal{M}|\log(|\mathcal{M}|))$. 
Moreover, since there are $|\mathcal{S}|$ subproblems, the time complexity of obtaining the value of $\sum_{\forall j \in \mathcal{S}}{U_j(\mathcal{M}_j)}$ in Step 2 in the algorithm is $O(|\mathcal{S}||\mathcal{M}|\log(|\mathcal{M}|))$.
Then, because the worst case of the algorithm is to traverse every possible combination, which is $|\mathcal{M}|^{|\mathcal{S}|}$, the time complexity of the algorithm is $O(|\mathcal{M}|^{|\mathcal{S}|}|\mathcal{S}||\mathcal{M}|\log(|\mathcal{M}|))$.
\end{proof}

\subsection{Submodular-Based Greedy Algorithm - ELVA}
In this section we introduce a $\rho(1-\frac{1}{e})$-approximation algorithm which we refer to as \textbf{E}fficient \textbf{L}ayered \textbf{V}ideo delivery \textbf{A}lgorithm (ELVA).
(Note that the value of $\rho$ is defined later in Lemma \ref{lemma_2}.)
Note that Problem (\ref{p1}) in Sec. \ref{Problem} is neither submodular nor monotone because of the term $\max_{\forall i \in \mathcal{M}}\left\{x_{i,j}N^b_{i,j}\right\}$ in the third constraint (i.e., multicasting the basic view). 
However, based on the layered structure of the problem setting (i.e., two-tier video), we can transform Problem (\ref{p1}) to a monotone submodular maximization problem
and use a standard greedy approach for monotone submodular maximization problems to solve it. 

\subsubsection{Problem Transformation}
We define the new monotone submodular maximization problem over a matroid constraint by replacing $U_j(\cdot)$ with $\widetilde{U}_j(\cdot)$ in Problem (\ref{bb}), where $\widetilde{U}_j(\cdot)$ is defined as follows:
\begin{align}
 {\widetilde{U}_j{(\mathcal{M}_j})} &~~ = \max_{y_{i,j,k}}{\sum_{\forall i \in \mathcal{M}_j}{\sum_{\forall k \in \mathcal{E}}{y_{i,j,k}w_{i,j,k}}}} \nonumber\\
\text{subject to:} &~~y_{i,j,k} \leq w_{i,j,k},  \;\; \forall i \in \mathcal{M}_j, \forall k \in \mathcal{E}, \nonumber\\
&~~\sum_{\forall i \in \mathcal{M}_j}\sum_{\forall k \in \mathcal{E}}{y_{i,j,k}N^e_{i,j,k}} \leq \widetilde{N}_j, \;\; \nonumber\\
&~~y_{i,j,k} \in [0,1], \;\; \forall i \in \mathcal{M}_j, \forall k \in \mathcal{E}. \nonumber
\end{align}
Note that $\widetilde{N}_j = N_j - \bar{N}^b$, where $\bar{N}^b$ is given by 
$\bar{N}^b = \frac{B}{R_{max\_min}}$,
and $R_{max\_min}$ is obtained by first finding 
$d_{max\_min} = \max_{i}{\{\min_j{\{d_{i,j}\}}\}}$
and then using Eq. (\ref{rate}) to find the corresponding $R_{max\_min}$.
Note that the above max-min problem can be solved in polynomial time by applying the divide-and-conquer algorithm for the closest pair of points problem \cite{Closest}.
Then, we have the following problem:
\begin{align}
\max_{x_{i,j}}&~~\sum_{\forall j \in \mathcal{S}}{\widetilde{U}_j(\mathcal{M}_j)} \nonumber \\
\text{subject to:} &~~\sum_{\forall j \in \mathcal{S}}x_{i,j} = 1, \forall i \in \mathcal{M}  \nonumber\\
&~~\mathcal{M}_j = \{i \in \mathcal{M} ~ | ~ x_{i,j} = 1\}  \nonumber\\
&~~x_{i,j} = \{0,1\}, \forall i \in \mathcal{M}, \forall j \in \mathcal{S}. \label{p3}
\end{align}

\subsubsection{Evaluation Function}
To greedily solve this problem as a monotone submodular maximization problem, we first introduce a so-called evaluation function. 
\begin{equation} \label{eva_func}
\Delta Q_{i,j}(\widehat{N}_j) = \min\{\bar{N}^b-N^b_{i,j}, 0\} \times T + G_{i,j}(\widehat{N}_j),
\end{equation}
where $T$ is an arbitrary big number, and $G_{i,j}(\widehat{N}_j)$ is obtained by solving the following linear programming, where $\widehat{N}_j$ is the current available RBs (this value will be updated at every iteration):
\begin{align}
G_{i,j}(\widehat{N}_j) &~~ = \max_{y_{i,j,k}}{\sum_{\forall k \in \mathcal{E}}{y_{i,j,k}w_{i,j,k}}} \nonumber\\
\text{subject to:} &~~\sum_{\forall k \in \mathcal{E}}{y_{i,j,k}N^e_{i,j,k}} \leq \widehat{N}_j, \;\; \nonumber\\
&~~y_{i,j,k} \in [0,1], \;\; \forall k \in \mathcal{E}.
\end{align}
We then use this function to select an element and update the value of the function iteratively until every element has been visited once.

\subsubsection{Algorithm}
The algorithm leverages the submodularity and monotonicity of Problem (\ref{p3}). 
With these two properties, we can directly apply the standard greedy algorithm for monotone submodular maximization problems, which always chooses the user-cell association with 
the maximal marginal value ($\Delta Q_{i,j}(\widehat{N}_j)$) based on the current available resources, see Algorithm \ref{ELVA} for more details.
\begin{algorithm}
	\caption{Submodular-Based Greedy Algorithm - ELVA}
	\begin{algorithmic}[1]
	\INPUT ${N^e_{i,j,k} \in \mathbb{R}}$, ${w_{i,j,k} = \{0,1\}}$, ${N_j \in \mathbb{R}}$, and ${N^b_{i,j} \in \mathbb{R}}$.
	\OUTPUT  $x_{i,j}$ and $y_{i,j,k}$.
	\State Initialize ${x_{i,j} = 0}$, ${y_{i,j,k} = 0}$, $\widehat{N}_j = N_j-\bar{N}^b$.
	\While {$\mathcal{M} \neq \emptyset$}
	\State Calculate $\Delta Q_{i,j}(\widehat{N}_j)$ for every $i$ and $j$.
	\State $x_{i,j} = 1$ with the highest value of  $\Delta Q_{i,j}(\widehat{N}_j)$.
	\State $i^* = i$ and $j^* = j$.
	\State Sort $\forall k \in \mathcal{E}$ by $\frac{w_{i^*,j^*,k}}{N^e_{i^*,j^*,k}}$ in descending order ($k^*$).
	\For {$\forall k^* \in \mathcal{E}$}
	\State $y_{i^*,j^*,k^*} =\min \left \{\frac{\widehat{N}_{j^*}}{N^e_{i^*,j^*,k^*}}, 1 \right \}$
	\State $\widehat{N}_{j^*} = \widehat{N}_{j^*} - y_{i^*,j^*,k^*} \times N^e_{i^*,j^*,k^*}$
	\EndFor
	\State $\mathcal{M} \leftarrow \mathcal{M} \setminus \{i^*\}$
	\EndWhile 
	\State Solve Problem (\ref{p2}) based on $x_{i,j}$.
	\end{algorithmic}
	\label{ELVA}
\end{algorithm} 

Note that since after Step 1 - Step 12 $\max\{x_{i,j}N^b_{i,j}\}$ might be smaller than $\bar{N}^b$, Step 13 helps allocate the remaining unused resources for maximizing the number of total rewards.
Also note that if there are two or more choices with the same value of $\Delta Q_{i,j}$, the algorithm will break the tie based on the best SINR.  

\subsubsection{Performance Analysis}
We establish an approximation ratio for ELVA.
First, we prove that Problem (\ref{p3}) is a monotone submodular maximization problem. 
Second, we prove that ELVA can always be better than any algorithms for Problem (\ref{p3}), and has the approximation ratio $\rho(1-\frac{1}{e})$. 
Then, we show that ELVA can be performed in polynomial time.

\begin{definition} \label{def_submodular}
(Submodularity \cite{Submodular}) A set function $f : 2^\mathcal{X} \rightarrow \mathbb{R}$ is submodular if, for all $\mathcal{A},\mathcal{B} \subseteq \mathcal{X}$ with $\mathcal{A} \subseteq \mathcal{B}$, and for all $i \in \mathcal{X} \setminus \mathcal{B}$,
$f(\mathcal{A} \cup \{i\}) - f(\mathcal{A}) \geq f(\mathcal{B} \cup \{i\}) - f(\mathcal{B})$.
\end{definition}

\begin{lemma} \label{thm_submodular}
$\widetilde{U}_j(\cdot)$ in Problem (\ref{p3}) is submodular.
\end{lemma}
\begin{proof}
To prove this, we need to show that $\widetilde{U}_j(\mathcal{A} \cup \{i\}) - \widetilde{U}_j(\mathcal{A}) \geq \widetilde{U}_j(\mathcal{B} \cup \{i\}) - \widetilde{U}_j(\mathcal{B})$ where $\mathcal{A} \subseteq \mathcal{B} \subseteq \mathcal{M}_j$ and $i \in \mathcal{M}_j \setminus \mathcal{B}$.
The proof is by illustrating all of the cases in the problem. 
Let $\widehat{N_j^{\mathcal{X}}} = \widetilde{N}_j - \sum_{i \in \mathcal{X}}\sum_{k \in \mathcal{E}}{y_{i,j,k}N^e_{i,j,k}}$. Given $\mathcal{A} \subseteq \mathcal{B}$, we have $\widehat{N_j^{\mathcal{A}}} \geq \widehat{N_j^{\mathcal{B}}}$. 
Case 1: Suppose $\widetilde{U}_j(\mathcal{A}) \leq \widetilde{U}_j(\mathcal{B})$. There exists an element $i \in \mathcal{M}_j \setminus \mathcal{B}$. If $i$ does not affect the ranking of both $\mathcal{A}$ and $\mathcal{B}$, we have $\widetilde{U}_j(\mathcal{A} \cup \{i\}) - \widetilde{U}_j(\mathcal{A}) = \widetilde{U}_j(\mathcal{B} \cup \{i\}) - \widetilde{U}_j(\mathcal{B}) = 0$. If $i$ affects the ranking of $\mathcal{A}$, but not the ranking of $\mathcal{B}$, we have $\widetilde{U}_j(\mathcal{A} \cup \{i\}) - \widetilde{U}_j(\mathcal{A}) \geq \widetilde{U}_j(\mathcal{B} \cup \{i\}) - \widetilde{U}_j(\mathcal{B}) = 0$. Note that the case when $i$ affects the ranking of $\mathcal{B}$, but not the ranking of $\mathcal{A}$ will not happen. 
Case 2: Suppose $\widetilde{U}_j(\mathcal{A}) > \widetilde{U}_j(\mathcal{B})$. We don't need to discuss this because it won't happen. 
Therefore, based on Definition \ref{def_submodular}, $\widetilde{U}_j(\cdot)$ in Problem (\ref{p3}) is submodular. 
\end{proof}

\begin{definition}
(Monotonicity \cite{Submodular}) A submodular function $f$ is monotone if for every $\mathcal{A} \subseteq \mathcal{B}$, we have that $f(\mathcal{A}) \leq f(\mathcal{B})$. 
\end{definition}

The following lemma follows directly from the fact that $\widetilde{U}_j(\cdot)$ is a linear function where $w_{i,j,k} \geq 0$.
\begin{lemma} \label{thm_monotone}
$\widetilde{U}_j(\cdot)$ in Problem (\ref{p3}) is monotone.
\end{lemma}

From the above, Problem (\ref{p3}) is a monotone submodular maximization with a matroid constraint or knapsack constraints \cite{Submodular}.
Let \textit{OPT1} be the optimal value of Problem (\ref{p1}) and \textit{OPT2} be the optimal value of Problem (\ref{p3}).

\begin{lemma} \label{lemma_1}
ELVA $\geq (1-\frac{1}{e})$OPT2.
\end{lemma}
\begin{proof}
Since ELVA uses the greedy algorithm for monotone submodular maximization problems in \cite{Submodular} and Problem (\ref{p3}) is a monotone submodular maximization problem from Lemma \ref{thm_submodular} and \ref{thm_monotone}, we can directly get the result from \cite{Submodular}. 
\end{proof}

\begin{lemma} \label{lemma_2}
$\rho$ OPT1 $\leq$ OPT2, where $\rho = \frac{\widetilde{N}_j}{N_j}$.
\end{lemma}
\begin{proof}
Let $y_{i,j,k}^1$ be the solution obtained by \textit{OPT1}. 
There are two things that need to be proved. First, we have to prove that $\rho \times y_{i,j,k}^1$ is also a solution of Problem (\ref{p3}) which is easy to see directly.
Second, we have to prove that $\max\{x_{i,j}N^b_{i,j}\}$ from \textit{OPT1} is smaller than or equal to $\bar{N}^b$ which follows directly considering how we obtain $\bar{N}^b$  
(see Problem Transformation subsection above). 
\end{proof}

From Lemmas \ref{lemma_1} and \ref{lemma_2} it follows directly that:
\begin{theorem} 
ELVA is a $\rho(1-\frac{1}{e})$-approximation algorithm.
\end{theorem}

\begin{lemma}
ELVA is a polynomial-time algorithm.
\end{lemma}
\begin{proof}
First, to calculate the evaluation function (\ref{eva_func}), it takes $O(|\mathcal{S}||\mathcal{M}|)$ computation. 
Then, to find the maximum value of (\ref{eva_func}), it requires the time complexity of $O(|\mathcal{S}||\mathcal{M}|\log(|\mathcal{S}||\mathcal{M}|))$.
The above two steps should be done for every mobile user (i.e., $|\mathcal{M}|$) so that the time complexity of ELVA is $O(|\mathcal{S}||\mathcal{M}|^2+|\mathcal{S}||\mathcal{M}|^2\log(|\mathcal{S}||\mathcal{M}|)).$ 
\end{proof}

\subsection{MKP-Based Greedy Algorithm - EVA}
The idea in this section is to use the greedy algorithm for MKP to solve our problem.
We first state the evaluation ratio to be used by the algorithm.

\subsubsection{Evaluation Ratio}
\begin{equation} \label{Eva_Ratio}
A_{i,j} = \frac{(\sum_{\forall k \in \mathcal{E}}w_{i,j,k})^p}{N^b_{i,j}}, 
\end{equation}
where $p$ is a term to tradeoff the importance of resource blocks represented by $N^b_{i,j}$ versus rewards represented by the weights $w_{i,j,k}$.
Note that in the performance evaluation section we will also consider the standard SINR-based greedy scheme, which can be obtained by simply setting $p$ to zero.

\subsubsection{Algorithm}
We introduce \textbf{E}fficient \textbf{V}ideo delivery \textbf{A}lgorithm (EVA), which is an algorithm that uses 
(\ref{Eva_Ratio}) to rank the choice for every user. Then, the algorithm greedily chooses user-association pairs based on the maximum value of (\ref{Eva_Ratio}) iteratively,
see Algorithm \ref{EVA} for more details.

\begin{algorithm}
	\caption{MKP-Based Greedy Algorithm - EVA}
	\begin{algorithmic}[1]
	\INPUT ${N^e_{i,j,k} \in \mathbb{R}}$, ${w_{i,j,k} = \{0,1\}}$, ${N_j \in \mathbb{R}}$, and ${N^b_{i,j} \in \mathbb{R}}$.
	\OUTPUT  $x_{i,j}$ and $y_{i,j,k}$.
	\State Initialize ${x_{i,j} = 0}$, ${y_{i,j,k} = 0}$, and $\check{N}^b_{j} = 0$.
	\State Sort ${\forall i \in \mathcal{S}}$ by ${A_{i,j}}$ in descending order.
	\For {${\forall i \in \mathcal{M}}$}
	\State $x_{i,j} = 1$ with the highest value of $A_{i,j}$.
	\State $\check{N}^b_{j} = \max\{N^b_{i,j}x_{i,j},\check{N}^b_{j}\}$
	\EndFor
	\For {${\forall j \in \mathcal{S}}$}
	\State $N_{j} = N_{j} - \check{N}^b_{j}$
	\EndFor
	\For {${\forall i \in \mathcal{M}}$}
	\State Sort $\forall k \in \mathcal{E}$ by $\frac{w_{i,j,k}}{N^e_{i,j,k}}$ in descending order.
	\For {$\forall k \in \mathcal{E}$}
	\State $y_{i,j,k} =\min \left \{\frac{N_{j}}{N^e_{i,j,k}}, 1 \right \}$
	\State $N_{j} = N_{j} - y_{i,j,k} \times N^e_{i,j,k}$
	\EndFor
	\EndFor
	\end{algorithmic}
	\label{EVA}	
\end{algorithm} 

\subsubsection{Performance Analysis}
The following result about EVA's approximation ratio follows directly from the results in \cite{GAP:Shmoys}.
\begin{theorem} 
EVA is a $\frac{1}{2}$-approximation algorithm.
\end{theorem}

\begin{lemma}
EVA is a polynomial-time algorithm.
\end{lemma}
\begin{proof}
To calculate the evaluation ratio (\ref{Eva_Ratio}) it takes $O(|\mathcal{S}||\mathcal{M}|)$ computation. 
Then, to find the maximum value of (\ref{Eva_Ratio}), the time complexity is $O(|\mathcal{S}||\mathcal{M}|\log(|\mathcal{S}||\mathcal{M}|))$.
The other steps require only linear time. Putting it all together, the time complexity of EVA is $O(|\mathcal{S}||\mathcal{M}|\log(|\mathcal{S}||\mathcal{M}|))$.
\end{proof}
\section{Performance Evaluation} 
\label{Simulation}
In this section, we compare different algorithms in a small-scale and a large-scale network to study their performance.
We use the following legends to label various algorithms: ``Optimal" for the optimal solution obtained by CVX \cite{CVX}, ``BB" for the Branch-and-Bound Method, ``ELVA" for the submodular-based greedy algorithm, ``EVA" for the MKP-based greedy algorithm (with a default value of $p$ equal to one), and ``SINR" for the SINR-based greedy algorithm.  

%
\begin{table} 
\caption{\textsc{Settings for Simulation Experiments}}
\vspace{-0.2cm}
\label{tab:small}
\begin{tabular}{|p{4.5cm}|l|l|}
\hline
\multirow{2}{*}{Description}    					& \multicolumn{2}{l|}{Notation} \\ \cline{2-3}
                                               					& Small-scale   & Large-scale  \\ \hline \hline
Default number of VR users      					&  ${|\mathcal{M}|=50}$ & ${|\mathcal{M}|=500}$  \\ \hline
Default number of small cells              					&  ${|\mathcal{S}|=10}$ & ${|\mathcal{S}|=100}$ \\ \hline
Default number of enhanced views              		&  ${|\mathcal{E}|=5}$ & ${|\mathcal{E}|=20}$ \\ \hline
Total number of RBs for small cell $j$ of each timeframe 	& \multicolumn{2}{l|}{${N_j = 50,000}$}         \\ \hline
Size of a basic view                           				& \multicolumn{2}{l|}{${B = 2}$ Mb}         \\ \hline
Size of an enhanced view                       			& \multicolumn{2}{l|}{${E_{k} = 2}$ Mb}         \\ \hline
Carrier frequency                              				& \multicolumn{2}{l|}{$f_c = 5$ GHz}         \\ \hline
Transmission power                            	 		& \multicolumn{2}{l|}{${P_t = 1}$ Watt}         \\ \hline
Noise power                                    				& \multicolumn{2}{l|}{$\mathcal{N} = -174$ dBm/Hz}         \\ \hline
Map Range                                      				& \multicolumn{2}{l|}{${1,000}$ m}  \\ \hline
\end{tabular}
\end{table}

\subsection{Simulation Settings}
We consider a network of small cells and VR users in a circle with a radius of 1,000 m. 
The path loss for the following simulations is based on WINNER-II model. According to this model, $g_{ij}$ is given in dB from the formula below:
\begin{equation}
g_{i,j}(d_{i,j}) = a \times \log_{10}(d_{i,j}) + b + c \times \log_{10}(f_c/5) + X,
\end{equation}
where $a$, $b$, $c$, and $X$ are parameters related with scenarios which can be found in \cite{WINNER-II}. 
The carrier frequency ($f_c$) is 5 GHz, the transmit power of a small cell ($P_t$) is 1 Watt, and the background noise power ($\mathcal{N}$) is assumed to be -174 dBm/Hz.
We assume every small cell has 50,000 RBs per second, where one RB is 180 KHz $\times$ 0.5 ms, and the total available bandwidth of the system is 100 MHz. 
The above parameters are set according to \cite{3GPP:36.814, Vook:Asilomar2018, Huang:ICC2016, Huang:TVT2017}. 
We assume that $w_{i,j,k} = 1, \forall i,j,k$, the number of enhanced views complementing a basic view $|\mathcal{E}|$ varies between 5 and 20, and the cache size $K$ default value is $\ceil{\frac{|\mathcal{E}|}{2}}$.
We apply a caching placement scheme based on \cite{FemtoCaching, Ao:TMC2018} to allocate enhanced views to small cell caches. The main idea of the scheme if to allocate enhanced views based on 
long term statistics about what users tend to request depending on their current position while guaranteeing that each enhanced view will be present in at least one cache. 
The size of a basic view is 2 Mb ($\sim$2.3Mb) and the size of an enhanced view is 2 Mb (for every $90^{\circ}$) based on \cite{Petrangeli:MM2017,Colman-Meixner:TB2019} and the recommendation of YouTube \cite{YouTube} for one-second 2K video and 4K 360 video, respectively.
Last, the default number of small cells is, enhanced views and VR users depends on the considered scenarios.
TABLE \ref{tab:small} summarizes the simulation parameters.

\begin{figure} 
\centering
  \begin{subfigure}[b]{0.5\columnwidth}
        \centering
        \includegraphics[scale=0.245]{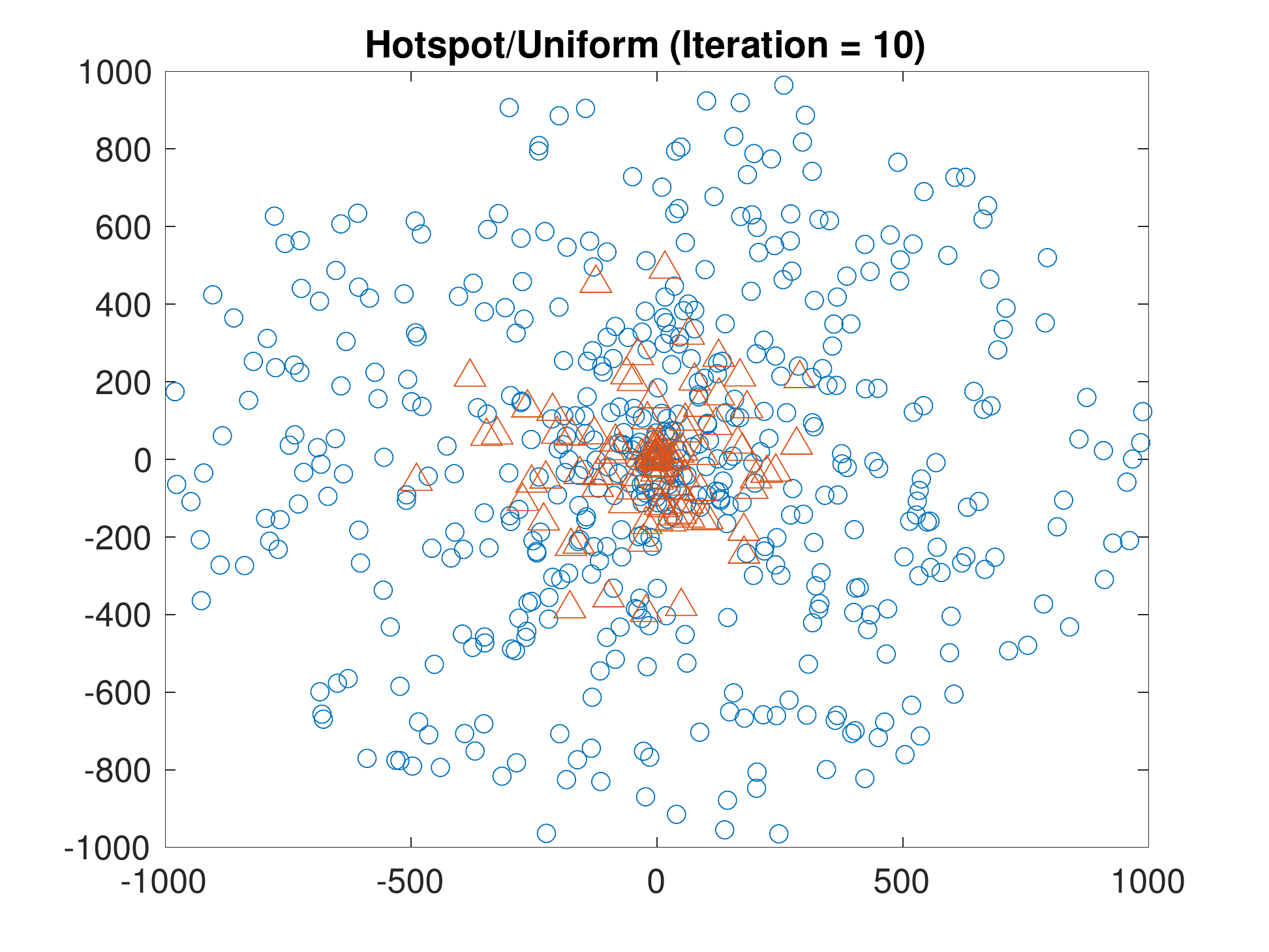}
        \caption{Hotspot/Uniform}
    \end{subfigure}%
    ~ 
    \hspace{-0.3cm}
    \begin{subfigure}[b]{0.5\columnwidth}
        \centering
        \includegraphics[scale=0.245]{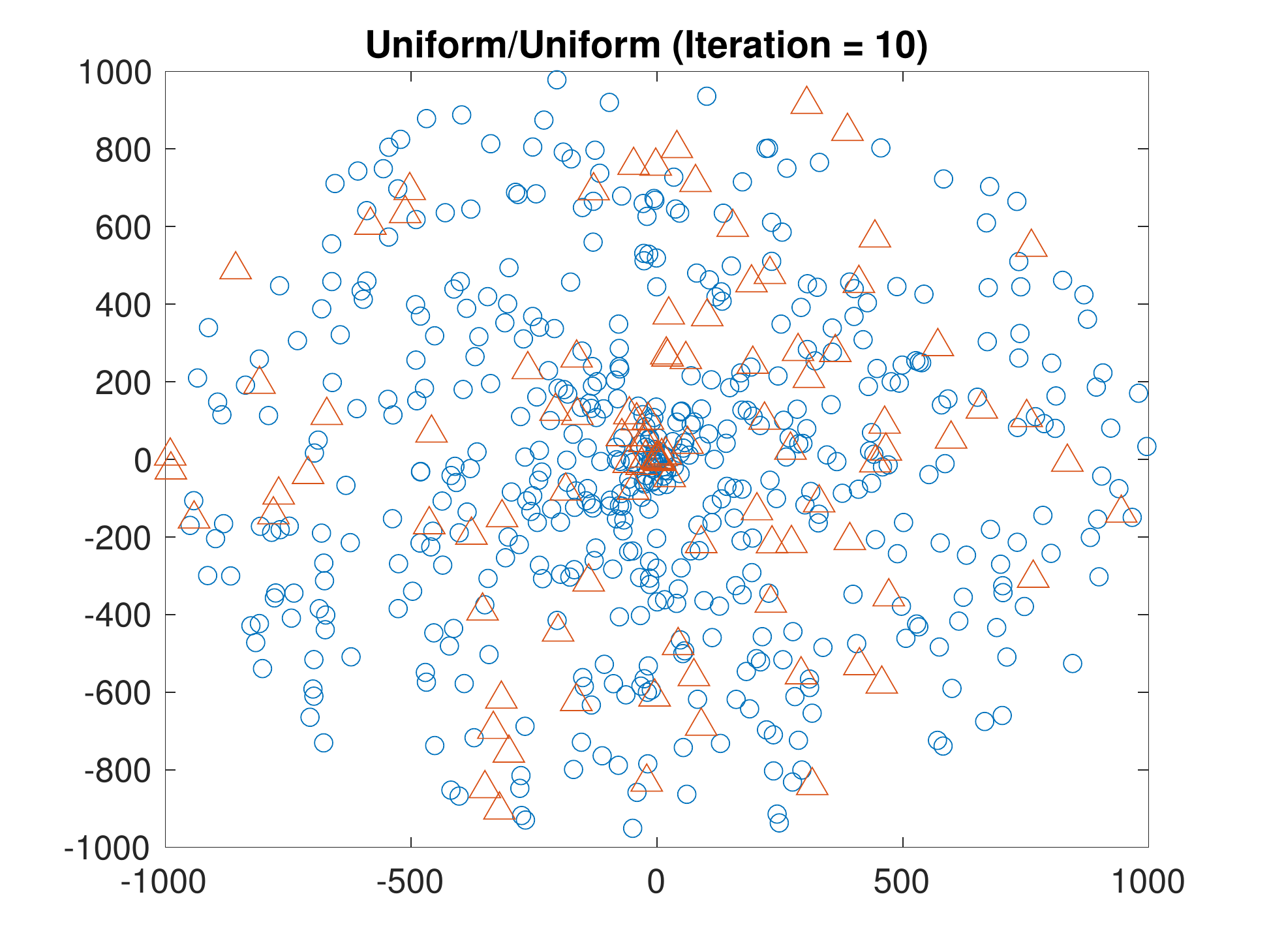}
        \caption{Uniform/Uniform}
    \end{subfigure}
\caption{Scenarios for small-scale study.}
\label{fig:scenarios}
\centering
\end{figure}
\begin{figure}
\centering
  \begin{subfigure}[b]{0.5\columnwidth}
        \centering
        \includegraphics[scale=0.24]{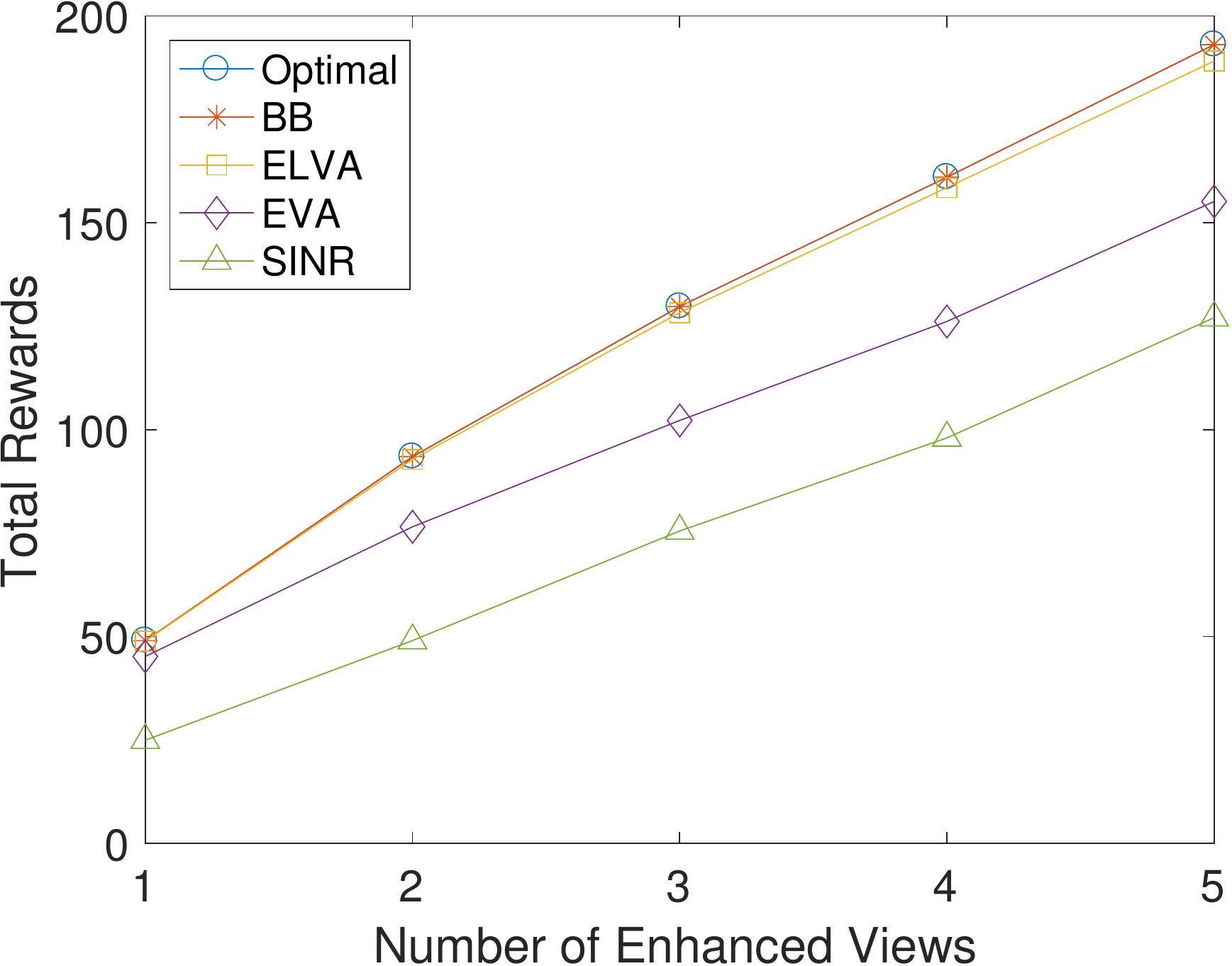}
        \caption{Hotspot/Uniform}
    \end{subfigure}%
    ~ 
    \begin{subfigure}[b]{0.5\columnwidth}
        \centering
        \includegraphics[scale=0.24]{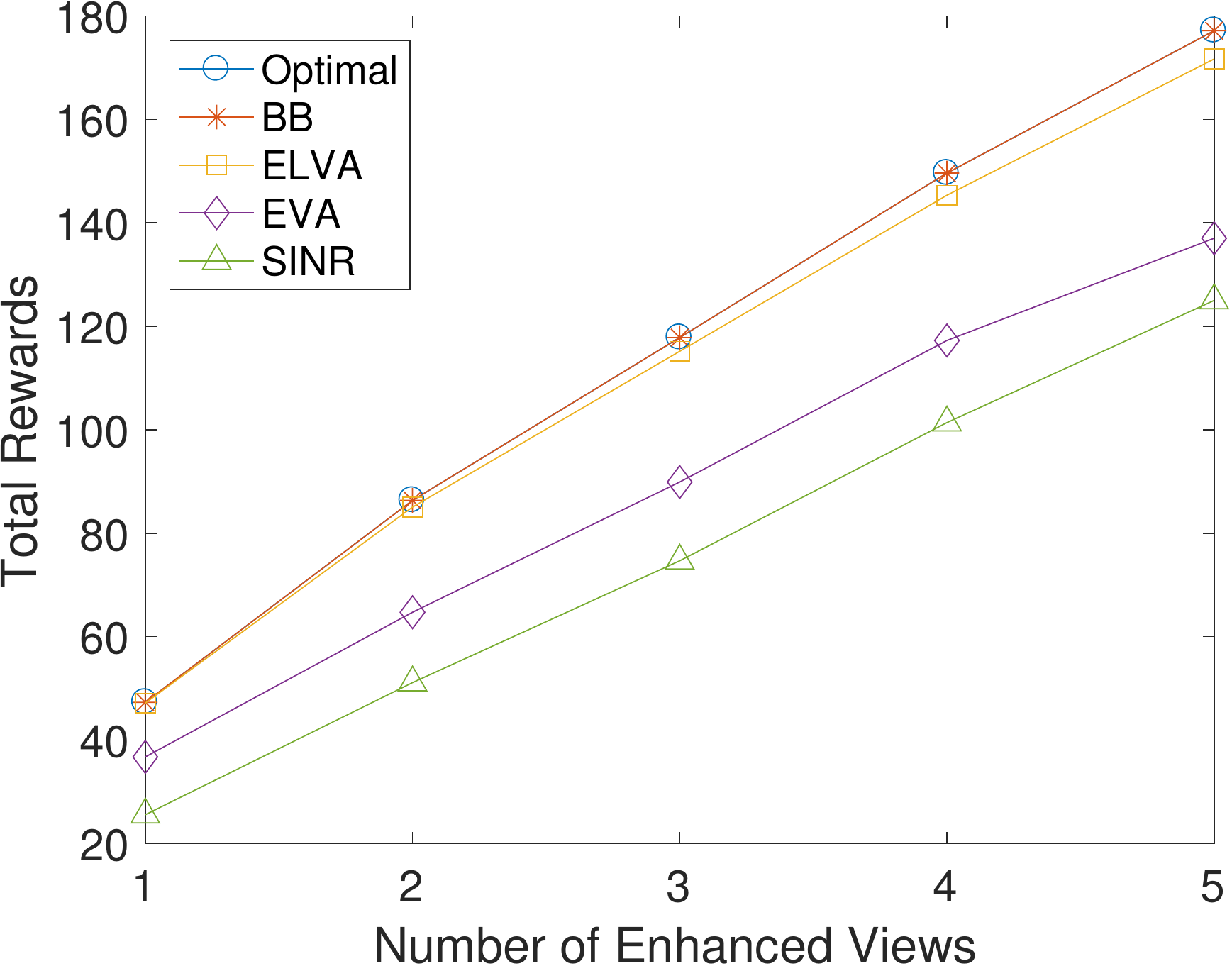}
        \caption{Uniform/Uniform}
    \end{subfigure}
\caption{Total rewards versus number of enhanced views.}
\label{fig:EVA_E}
\centering
\end{figure}
\begin{figure}
\centering
  \begin{subfigure}[b]{0.5\columnwidth}
        \centering
        \includegraphics[scale=0.24]{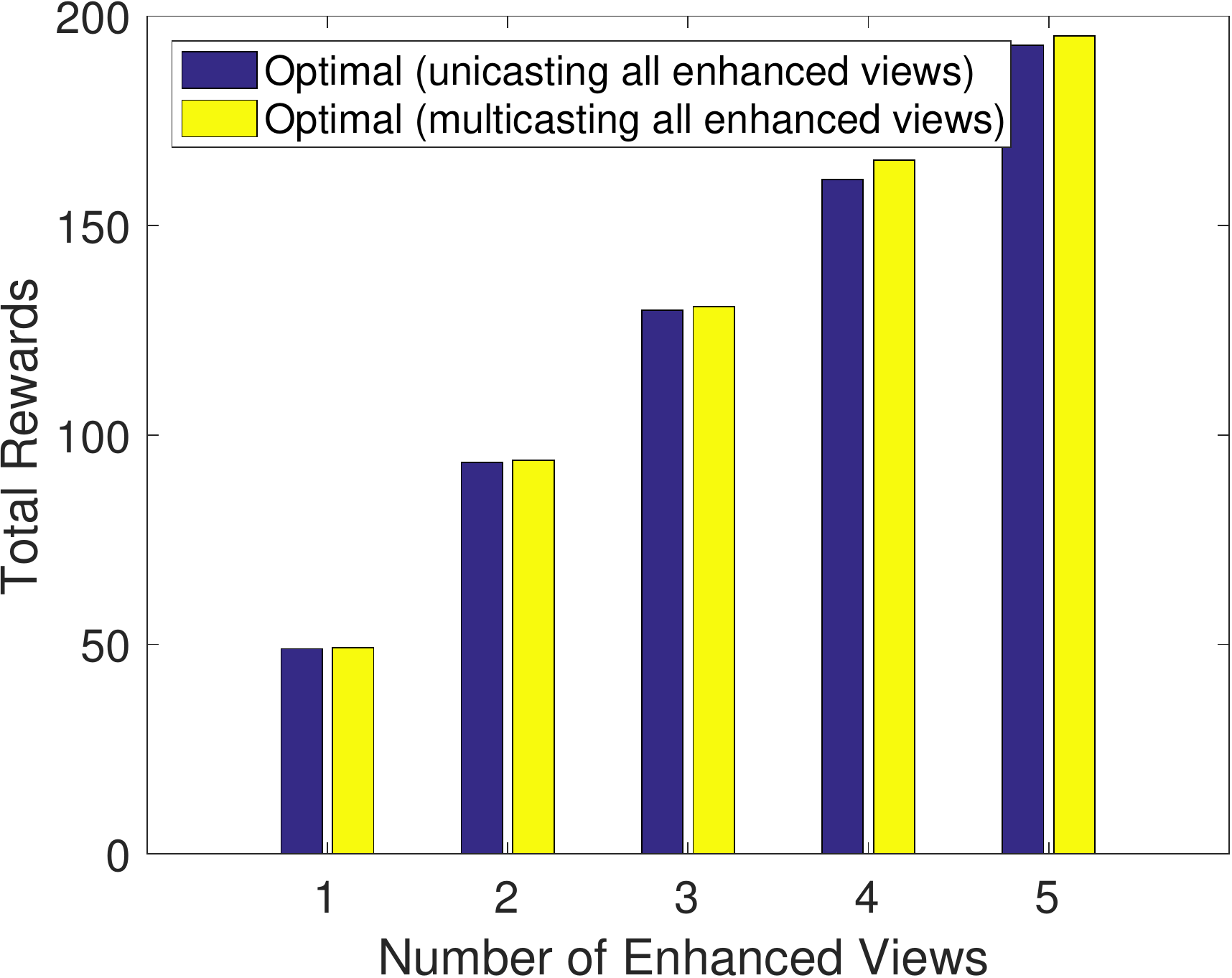}
        \caption{Hotspot/Uniform}
    \end{subfigure}%
    ~ 
    \begin{subfigure}[b]{0.5\columnwidth}
        \centering
        \includegraphics[scale=0.24]{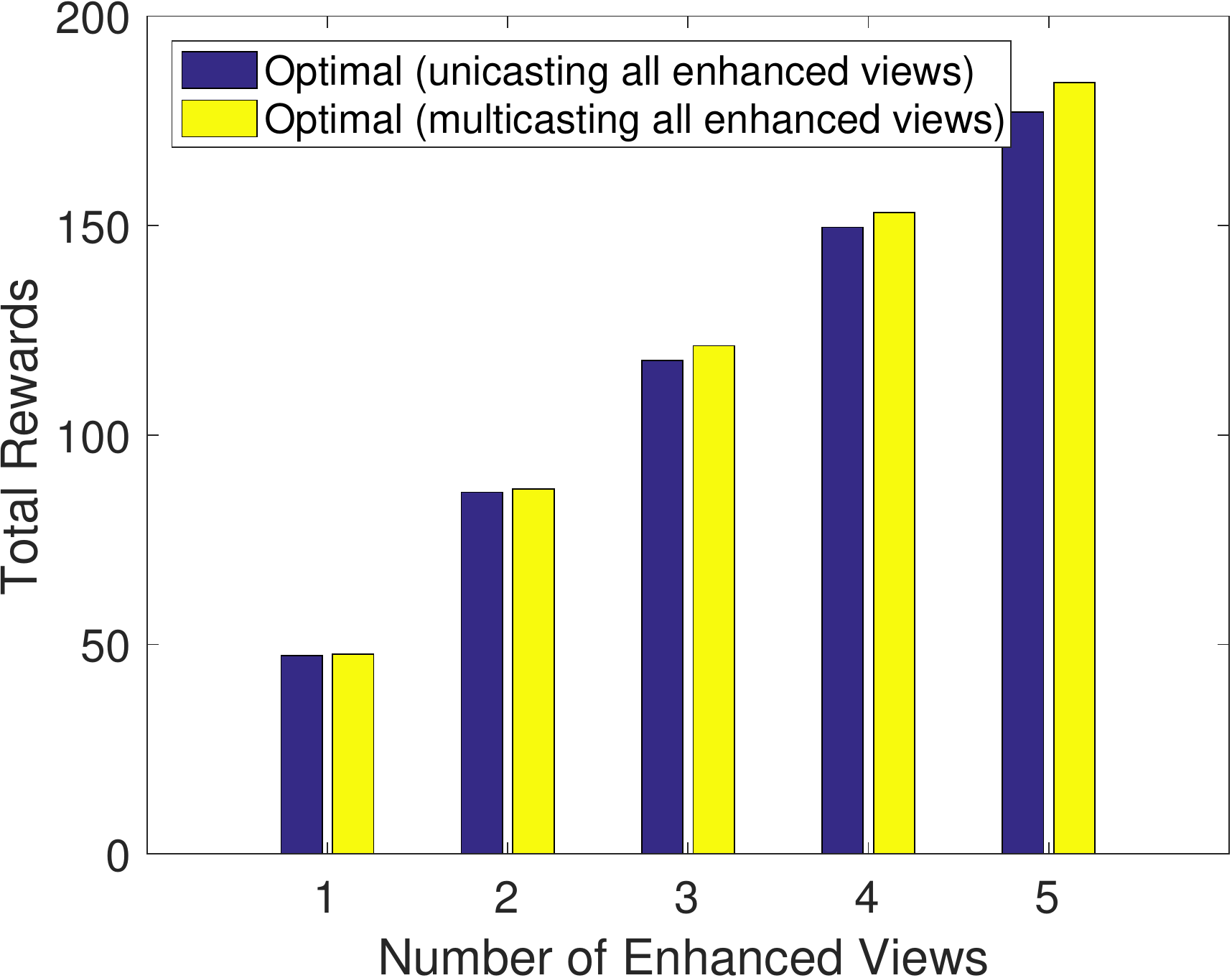}
        \caption{Uniform/Uniform}
    \end{subfigure}
\caption{Comparison of unicasting and multicasting enhanced views. (Sec. \ref{Multicasting})}
\label{fig:EVA_E2}
\centering
\end{figure}
\begin{table} [t]
\centering
\caption{Number of Small Cells with \% of Total RBs}
\vspace{-0.2cm}
\label{tab:Utilization}
\begin{tabular}{|l|l|l|l|l|l|}
\hline
\textbf{RU (\%)} & \textbf{0-20\%} & \textbf{21-40\%} & \textbf{41-60\%} & \textbf{61-80\%} & \textbf{81-100\%} \\ \hline \hline
Optimal/BB & 0 & 0 & 0 & 0 & 10 \\ \hline
ELVA & 2 & 3 & 2 & 1 & 2 \\ \hline
EVA & 9 & 1 & 0 & 0 & 0 \\ \hline
SINR & 10 & 0 & 0 & 0 & 0 \\ \hline
\end{tabular}
\end{table}

\begin{figure} [t]
\centering
  \begin{subfigure}[b]{0.15\textwidth}
        \centering
        \includegraphics[scale=0.15]{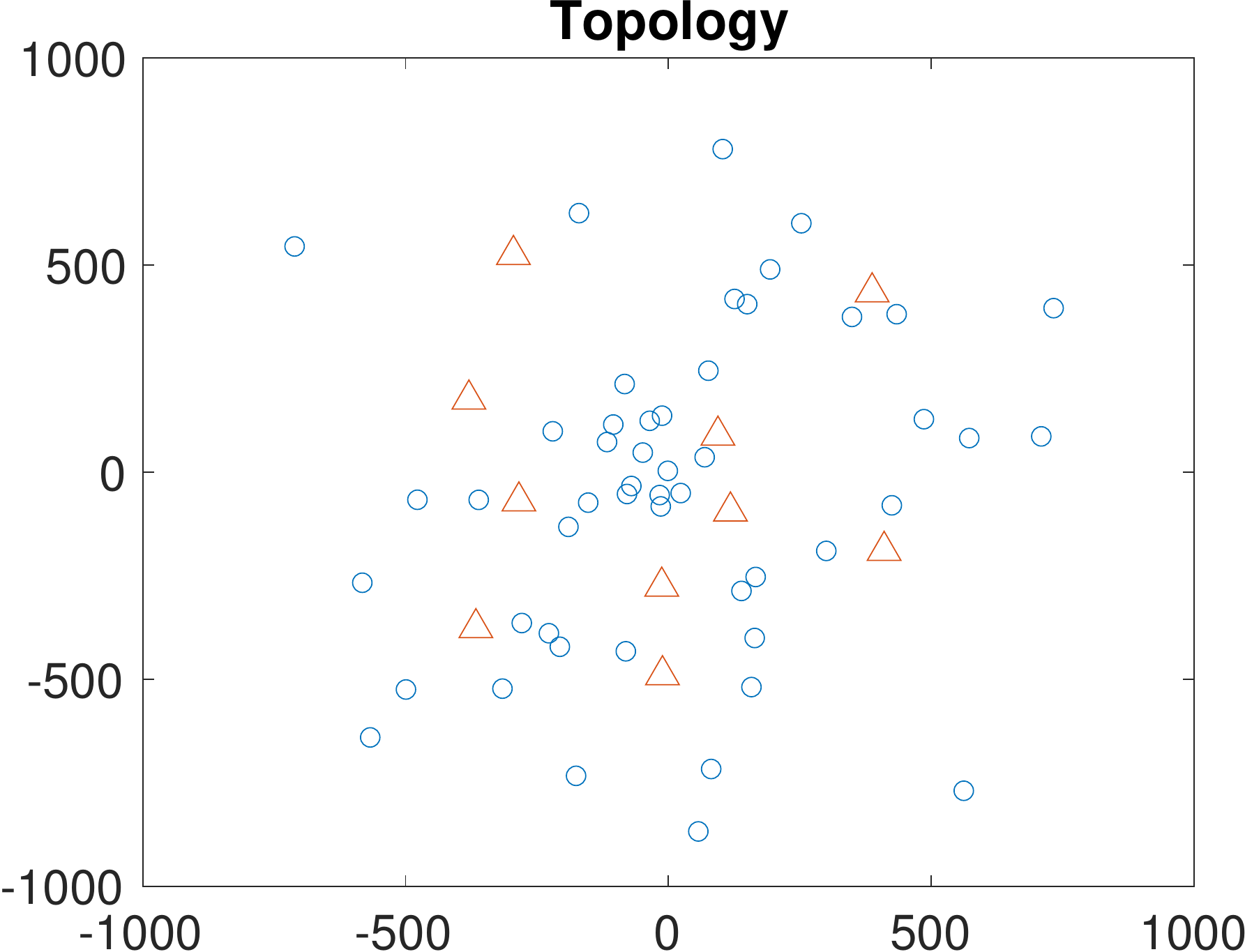}
        \caption{Topology}
    \end{subfigure}%
    ~
    \begin{subfigure}[b]{0.15\textwidth}
        \centering
        \includegraphics[scale=0.15]{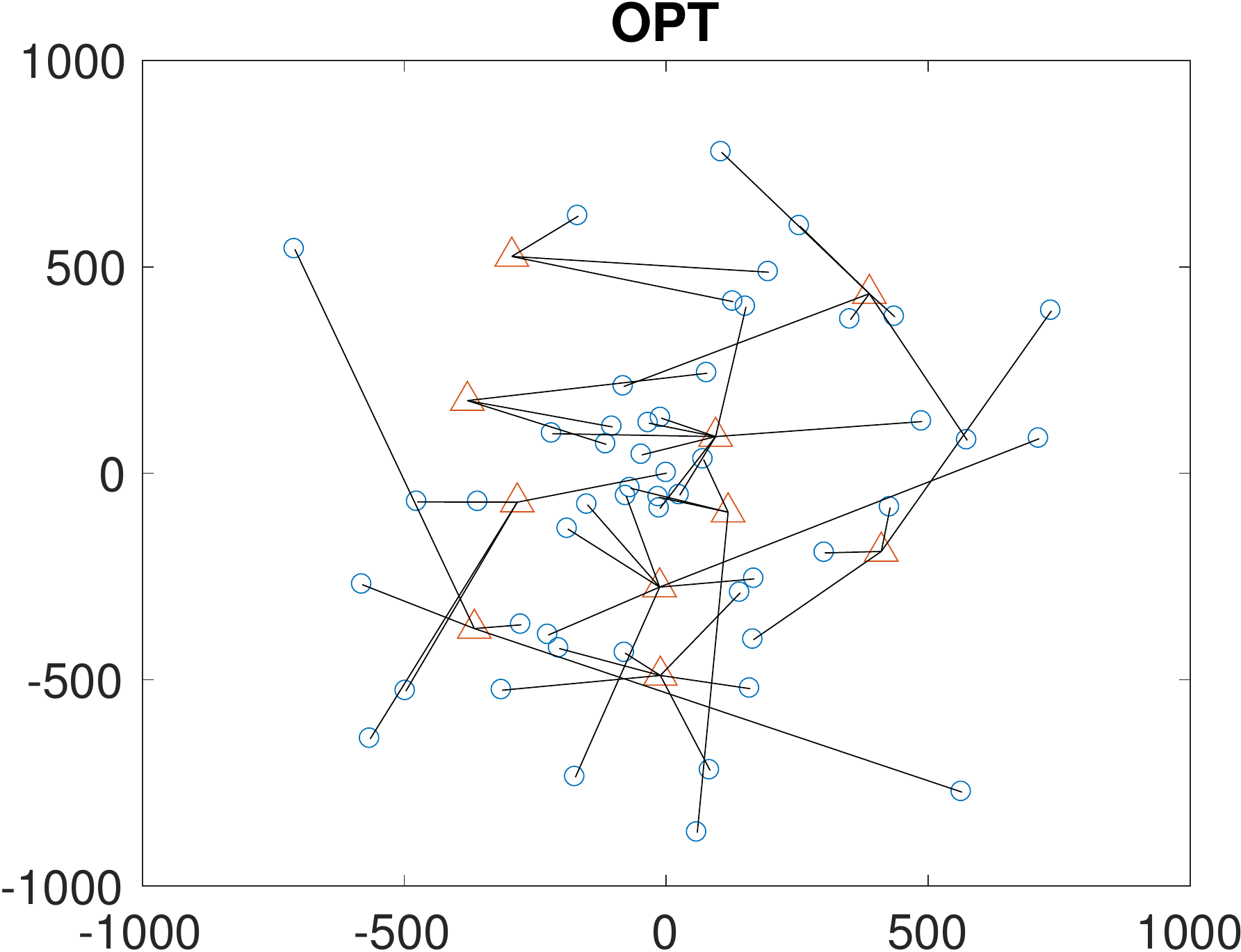}
        \caption{Optimal/BB}
    \end{subfigure}
    ~
    \begin{subfigure}[b]{0.15\textwidth}
        \centering
        \includegraphics[scale=0.15]{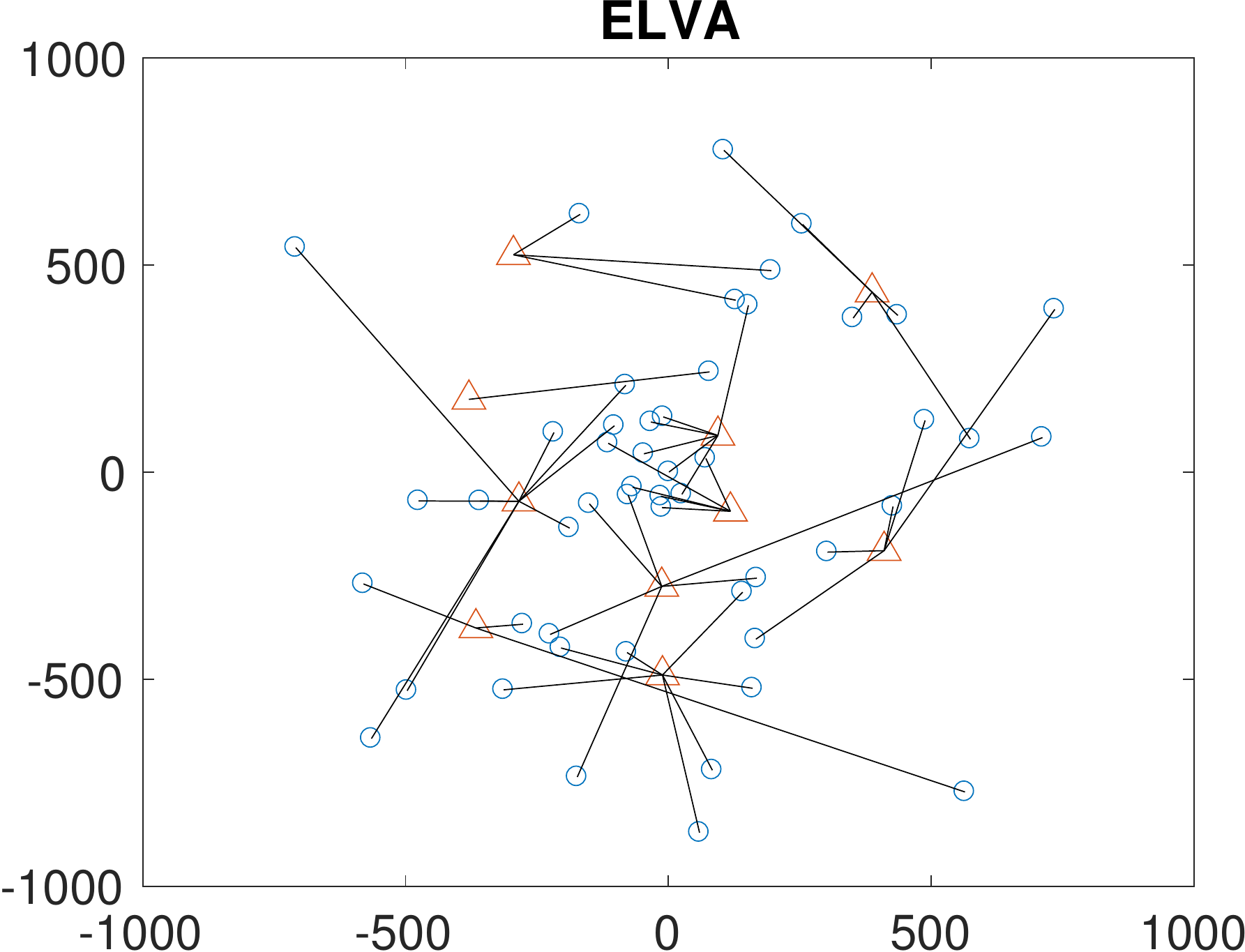}
        \caption{ELVA}
    \end{subfigure}
    ~
    \begin{subfigure}[b]{0.15\textwidth}
        \centering
        \includegraphics[scale=0.15]{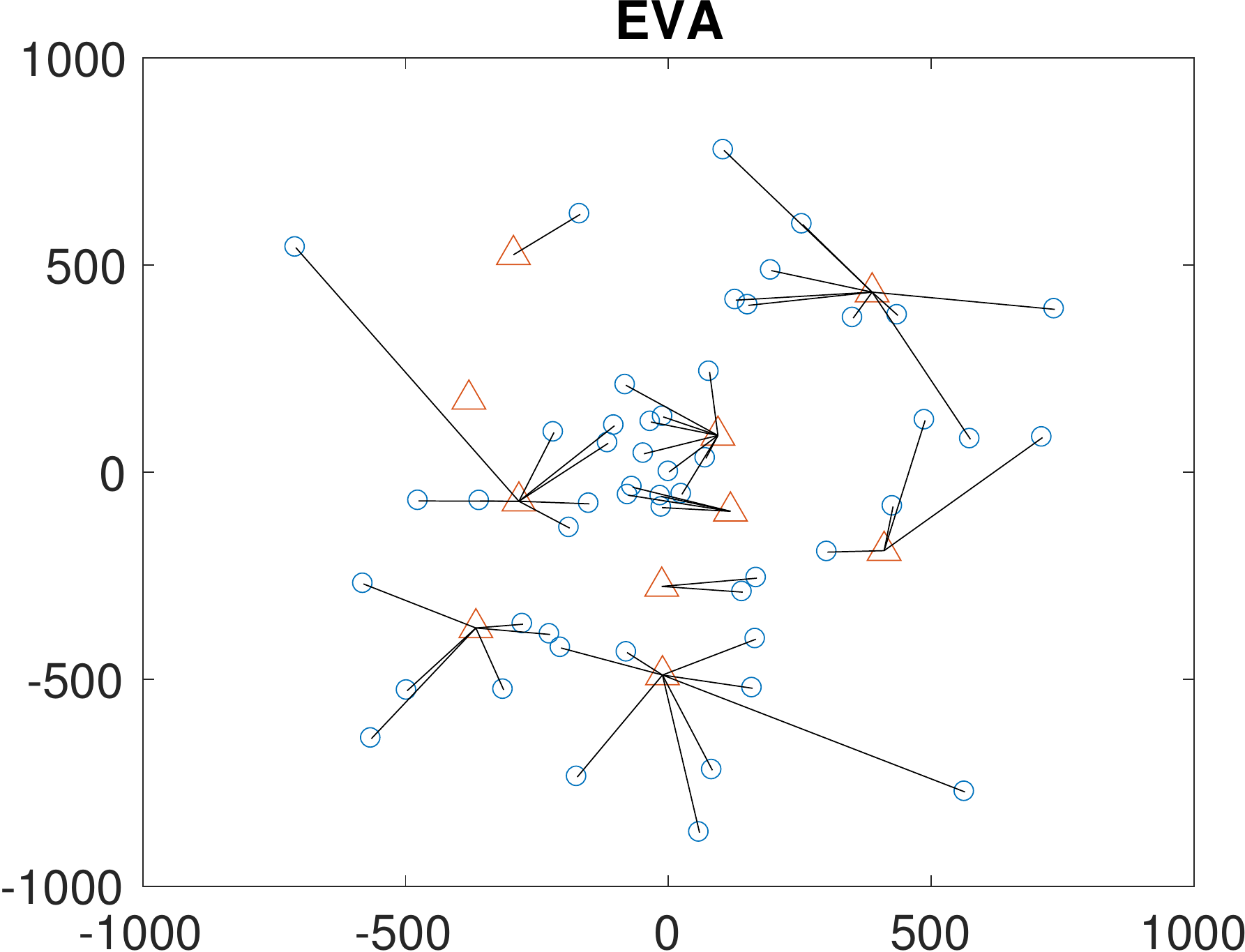}
        \caption{EVA}
    \end{subfigure}
    ~
    \begin{subfigure}[b]{0.15\textwidth}
        \centering
        \includegraphics[scale=0.15]{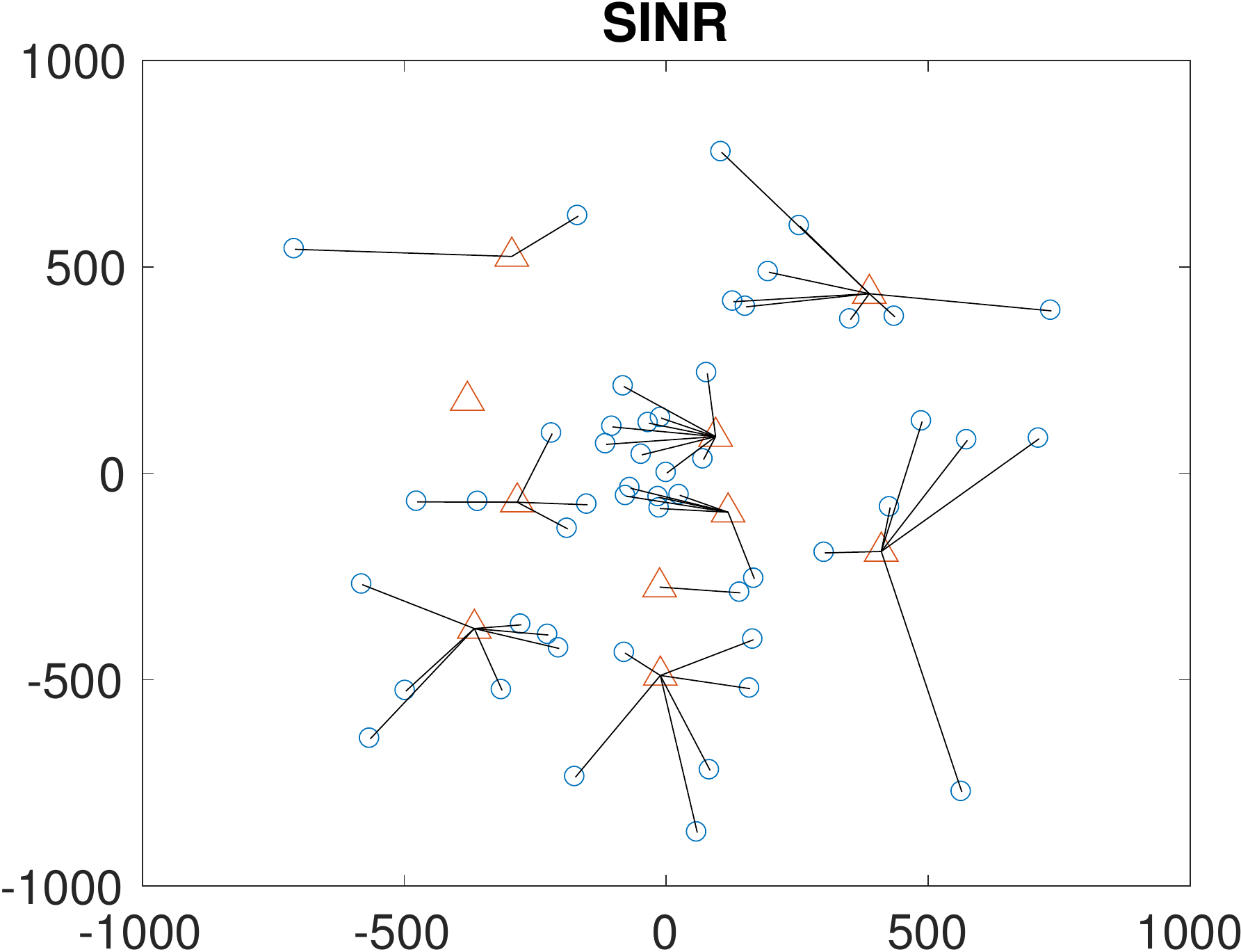}
        \caption{SINR}
    \end{subfigure}
\caption{User association in Uniform/Uniform scenario.}
\label{fig:Association}
\centering
\end{figure}

\subsection{Simulation Results - Small-Scale Study}
For the small-scale setup the default number of small cells is 10, the default number of enhanced views is 5, and the default number of VR users is 50. (Thus, on average, there are about 5 users per small cell, and hence 25 enhanced views of interest to those 5 users. As a result, each small cell cache would store on average about 10\% of the enhanced views of interest.) 

We study two scenarios in the small-scale network:
(i) small cells are ``normally" distributed with mean of $(0,0)$ and variance of 200 m in the given area, and VR users are uniformly distributed in the given area (``Hotspot/Uniform"), see Fig. \ref{fig:scenarios} (a)), and
(ii) both small cells and VR users are uniformly distributed in the given area (``Uniform/Uniform"), see Fig, \ref{fig:scenarios} (b)).

\subsubsection{Varying Number of Enhanced Views}
Fig. \ref{fig:EVA_E} and \ref{fig:EVA_E2} plot the total rewards as the number of enhanced views varies from 1 to 5.  
In Fig. \ref{fig:EVA_E}, when the number of enhanced views increases, the total rewards of all algorithms increase as expected. 
``Optimal" and ``BB" have the same performance, which validates the optimality of the Branch-and-Bound Method.
``ELVA" is near-optimal, with only 1\%$\sim$3\% difference from ``Optimal".
``EVA" has $\sim$20\% difference from ``Optimal", and ``SINR" has $\sim$40\% difference from ``Optimal.
In Fig. \ref{fig:EVA_E} (a) and (b), the difference between ``EVA" and ``SINR" in ``Uniform/Uniform" is smaller than the difference in ``Hotspot/Uniform", 
since the SINR values of cell-user pairs are closer in the former case.
Fig. \ref{fig:EVA_E2} depicts the advantage of multicasting enhanced views (see Sec. \ref{Multicasting}). 
In addition, as expected, when the number of enhanced views increases, the benefit becomes noticeable. 

\subsubsection{User Association and Resource Utilization} \label{details}
We plot the resulting user association and resource utilization when using each algorithm to understand the reason for their performance gap.
We measure resource utilization (RU) by the number of used RBs over the number of total RBs in a small cell. 
First, we plot the user association of every algorithm under the ``Uniform/Uniform" scenario in Fig. \ref{fig:Association}.
For reference only, Fig. \ref{fig:Association} (a) depicts the topology before we perform any algorithms.
As shown in Fig. \ref{fig:Association} (b) and (c), ``ELVA" makes similar choices for user association as ``Optimal/BB", explaining why ``ELVA" has a near-optimal performance. 
Moreover, in Fig. \ref{fig:Association} (d) and (e), ``EVA" makes similar choices for user association as ``SINR" which connects users to the nearest cell, and, consistent with the discussion
in the previous section, ``EVA" and ``SINR" have similar performance in the uniform topology of small cells.

Second, we present the number of small cells with different levels of RU in TABLE \ref{tab:Utilization}. 
All of the small cells in ``Optimal/BB" fully utilize their resources as expected, whereas the small cells in ``ELVA" partially utilize their resources,
because some mobile users of ``ELVA" choose different small cells to obtain the enhanced views. 
Moreover, the small cells in ``EVA" and ``SINR" underutilize their resources depending on their user-association. 
Although their mobile users connect to the nearest small cells, these mobile users cannot obtain more enhanced views from the small cells. 

\subsubsection{Varying Number of Small Cells}
Fig. \ref{fig:EVA_S} plots the total rewards versus the number of small cells which varies from 1 to 10. 
As expected, when the number of small cells increases, the total rewards of ``Optimal", ``BB", ``ELVA", and ``EVA" increase. 
``SINR" stays stable, since it only considers distance as a metric to make association decisions.

\begin{figure} 
\centering
  \begin{subfigure}[b]{0.5\columnwidth}
        \centering
        \includegraphics[scale=0.24]{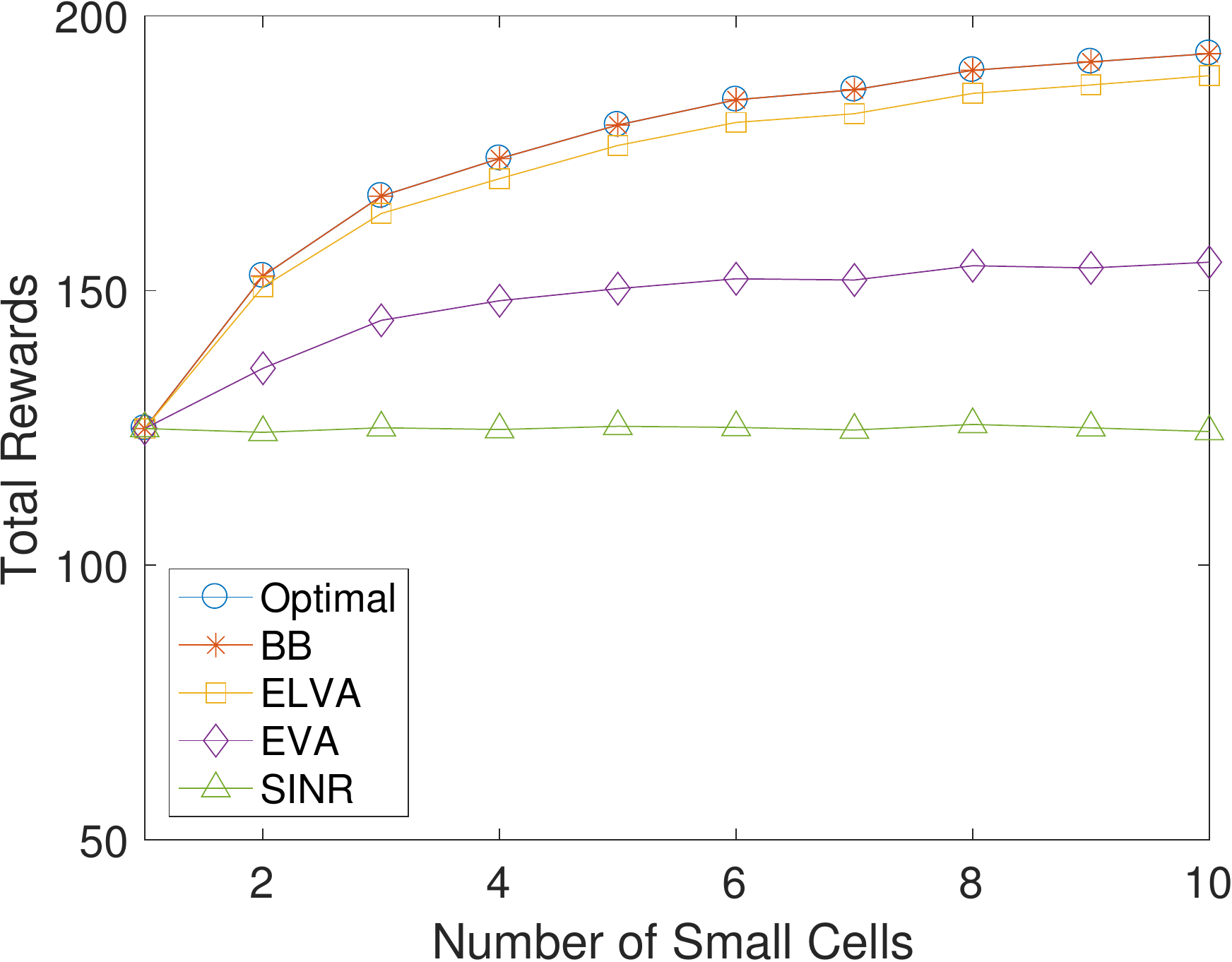}
        \caption{Hotspot/Uniform}
    \end{subfigure}%
    ~ 
    \begin{subfigure}[b]{0.5\columnwidth}
        \centering
        \includegraphics[scale=0.24]{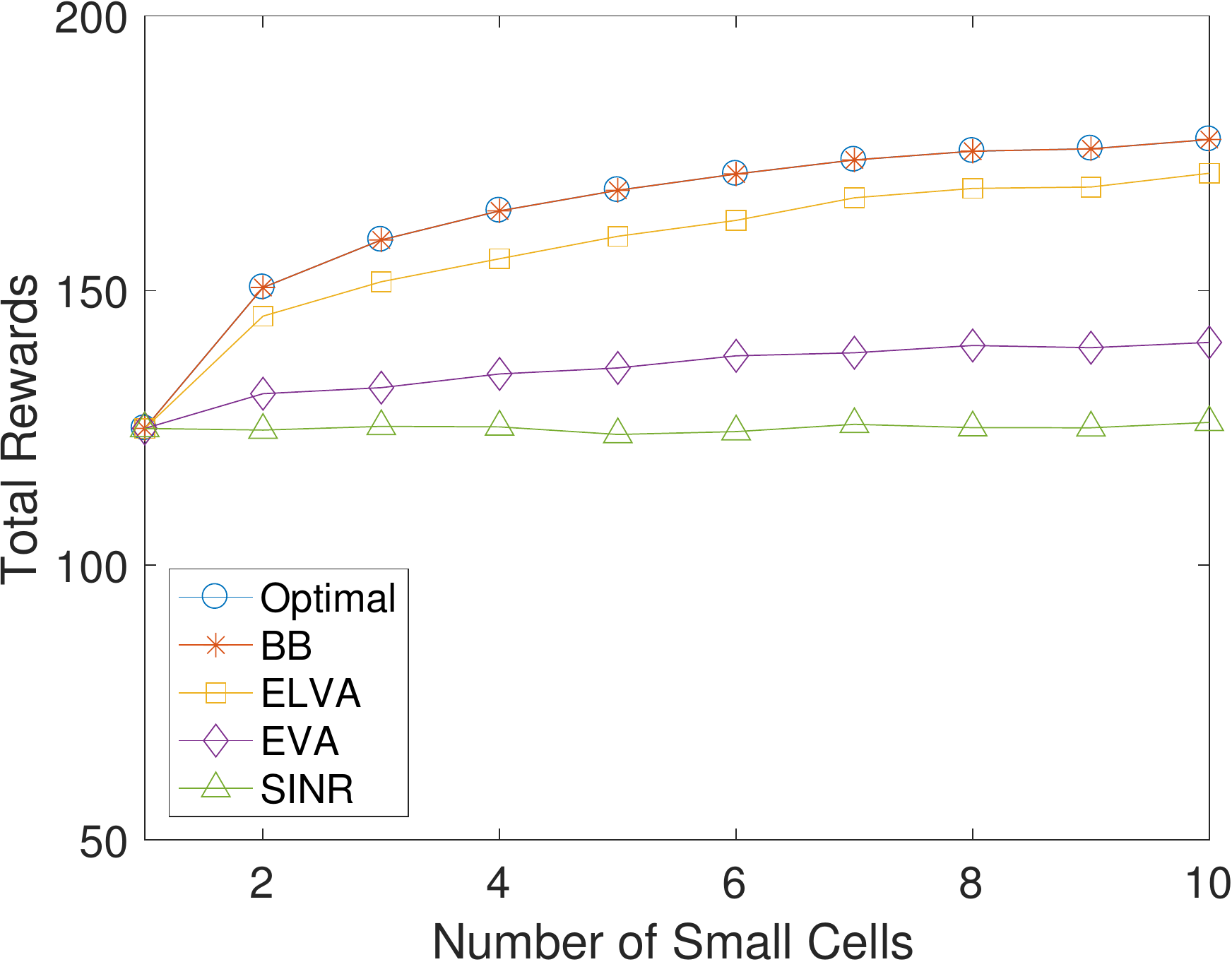}
        \caption{Uniform/Uniform}
    \end{subfigure}
\caption{Total rewards versus number of small cells.}
\label{fig:EVA_S}
\centering
\end{figure}

\subsubsection{Varying Number of Mobile Users}
In Fig. \ref{fig:EVA_M}, the total rewards are depicted versus the number of mobile users which varies from 10 to 50. 
As expected, when the number of mobile users increases, the total rewards of all algorithms increase. 

\begin{figure} 
\centering
  \begin{subfigure}[b]{0.5\columnwidth}
        \centering
        \includegraphics[scale=0.24]{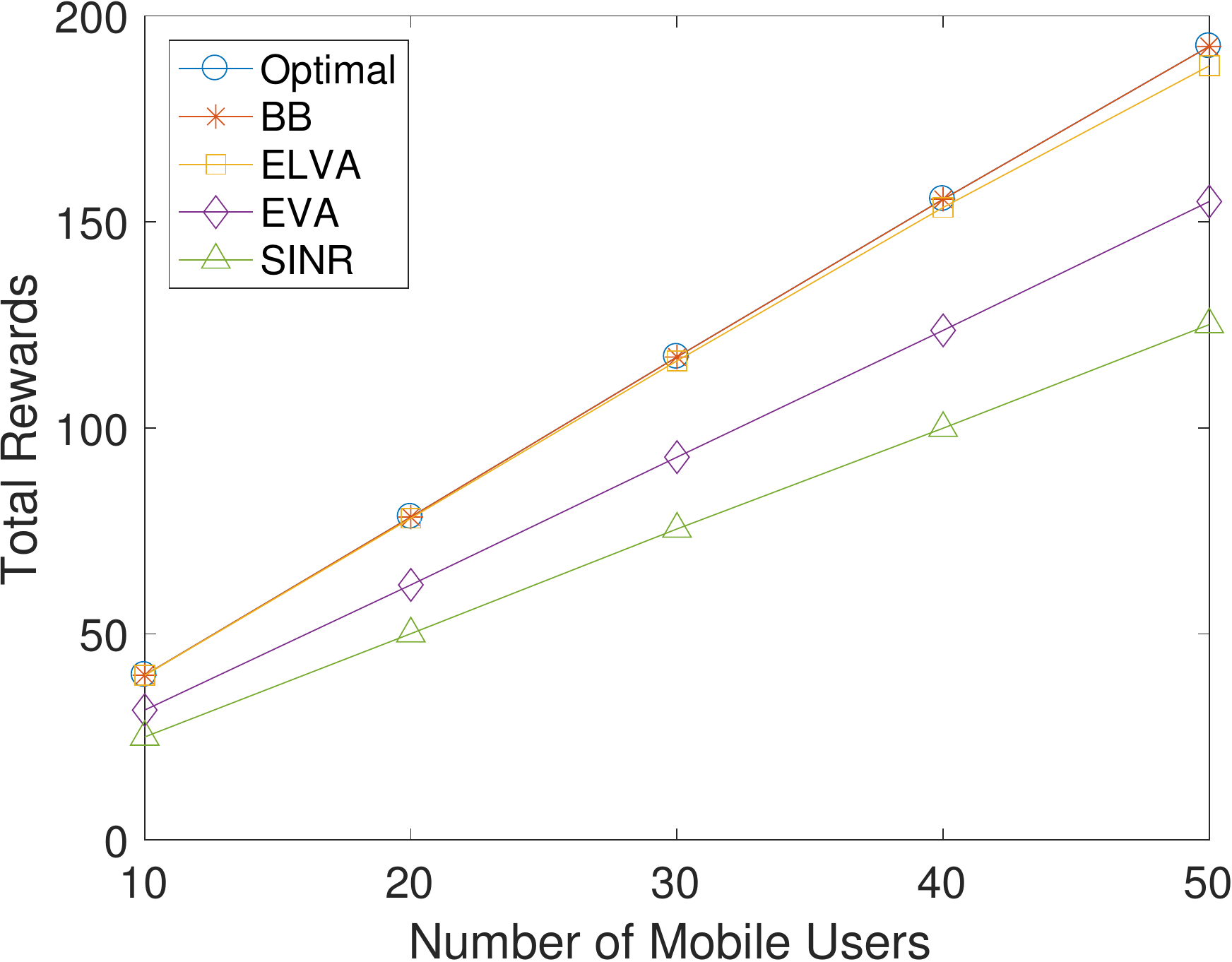}
        \caption{Hotspot/Uniform}
    \end{subfigure}%
    ~ 
    \begin{subfigure}[b]{0.5\columnwidth}
        \centering
        \includegraphics[scale=0.24]{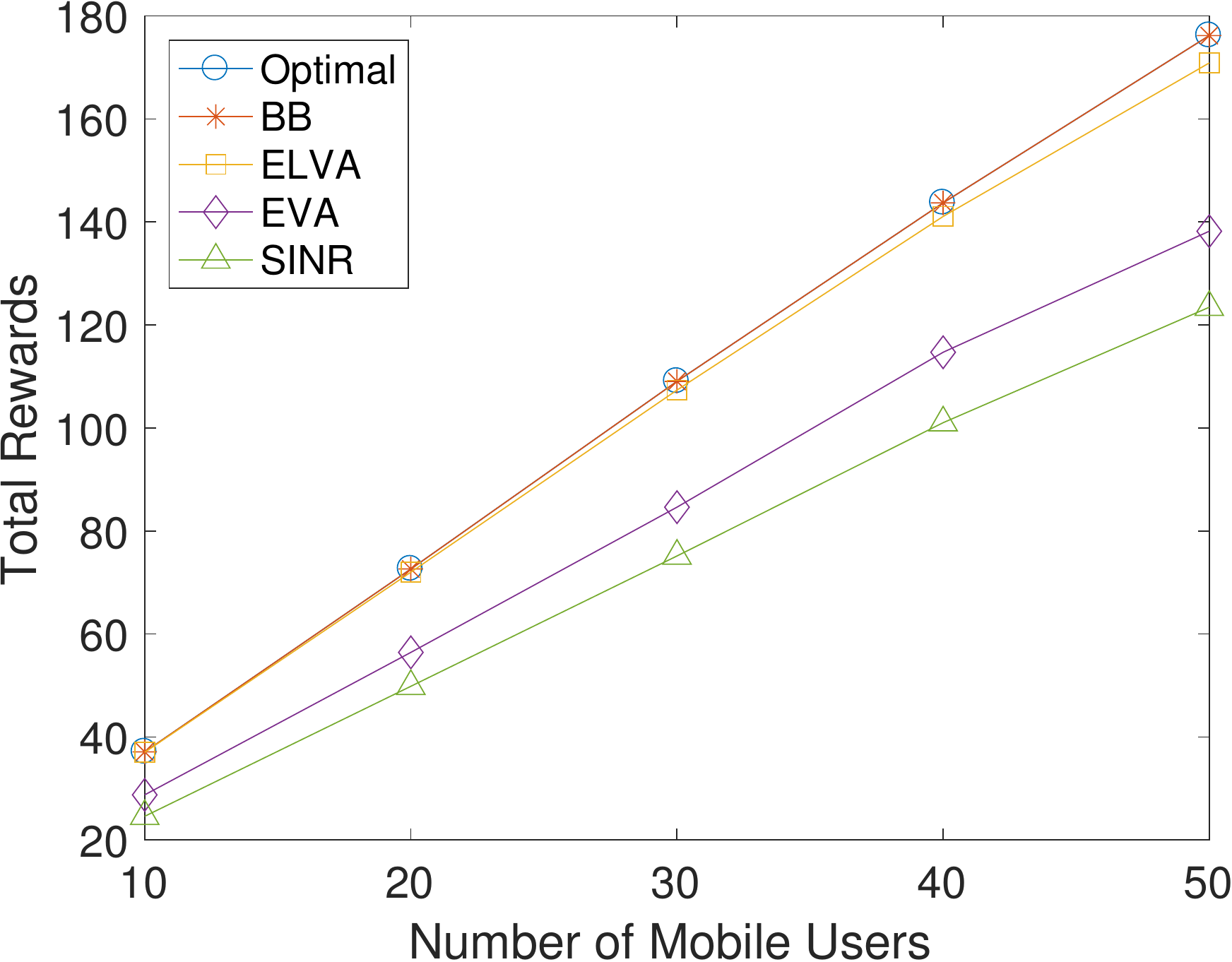}
        \caption{Uniform/Uniform}
    \end{subfigure}
\caption{Total rewards versus number of mobile users.}
\label{fig:EVA_M}
\centering
\end{figure}

\subsubsection{Varying Parameter of $p$ in EVA}
In Fig. \ref{fig:EVA_p}, we plot the total rewards versus the $p$ value in (\ref{Eva_Ratio}) which varies between 1 and 5. 
Note that with the larger $p$ value, ``EVA" puts more emphasis on the ``weight" than on SINR. 
We observe that as the $p$ value increases, the performance of ``EVA" is a concave function, which implies that
focusing on the rewards more and more (At the expense of SINR) doesn't necessarily improve performance.
More specifically, the best value of $p$ is 3 in ``Hotspot/Uniform, and the best value of $p$ is 4 in ``Uniform/Uniform". 

\begin{figure} 
\centering
  \begin{subfigure}[b]{0.5\columnwidth}
        \centering
        \includegraphics[scale=0.24]{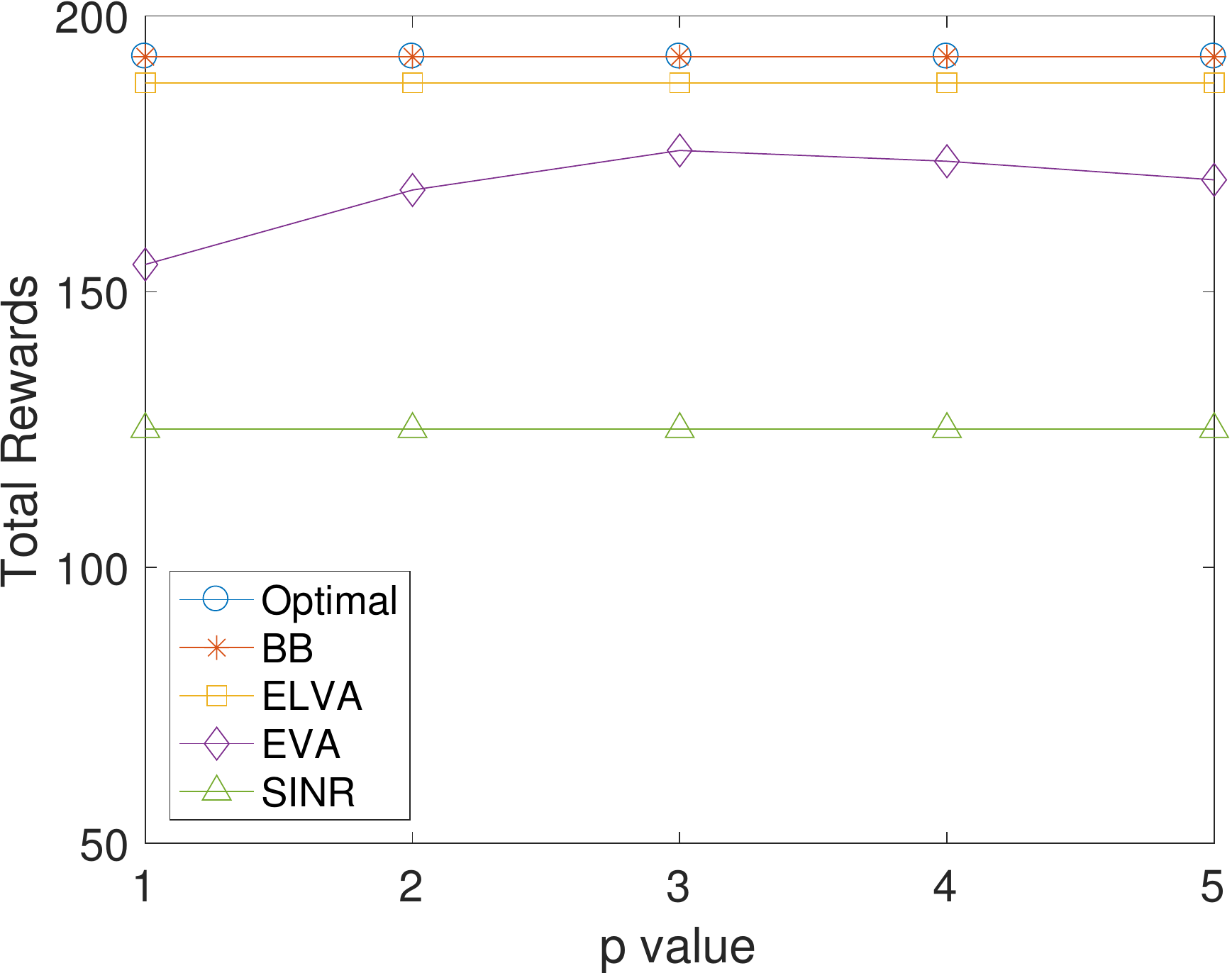}
        \caption{Hotspot/Uniform}
    \end{subfigure}%
    ~ 
    \begin{subfigure}[b]{0.5\columnwidth}
        \centering
        \includegraphics[scale=0.24]{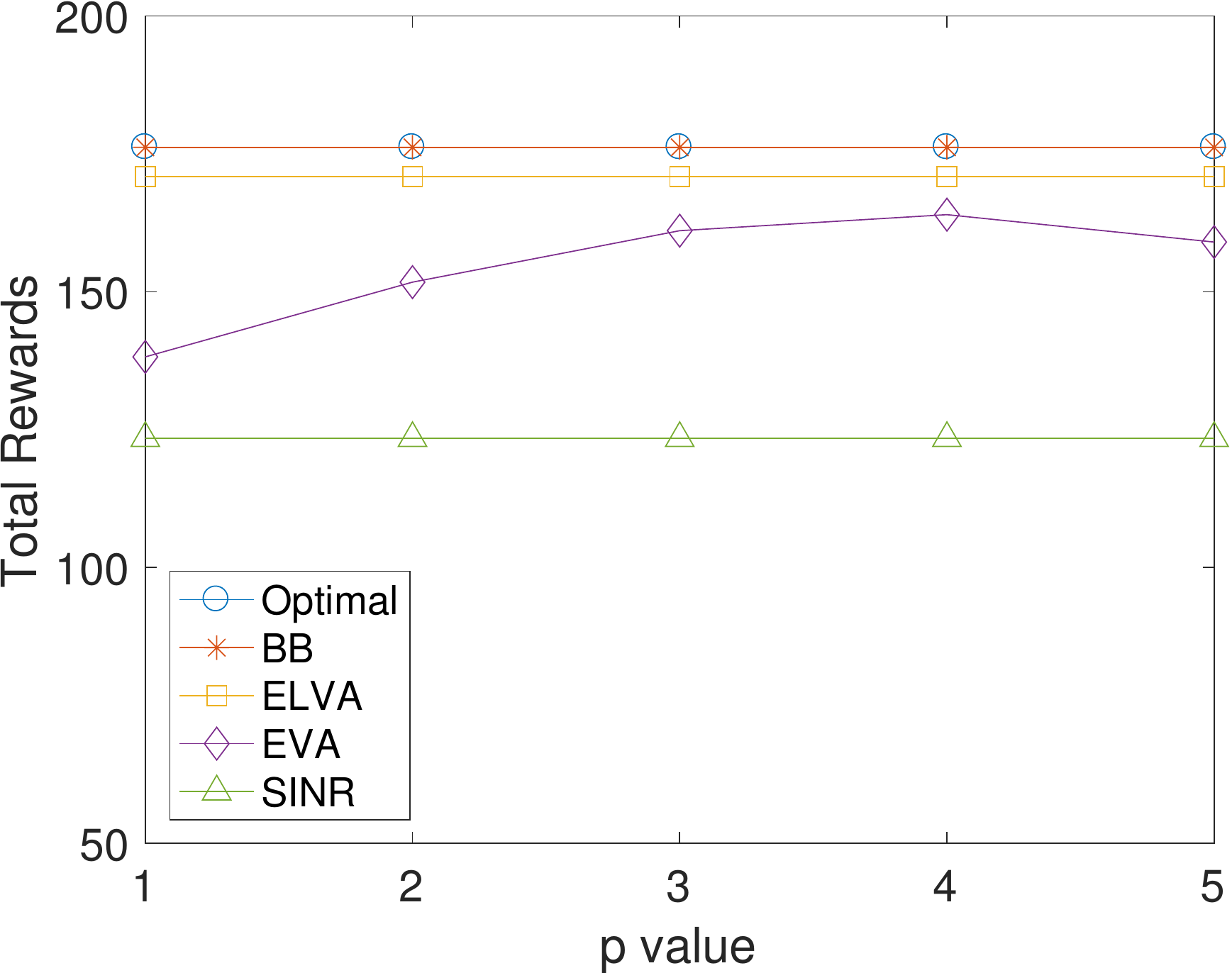}
        \caption{Uniform/Uniform}
    \end{subfigure}
\caption{Total rewards versus $p$ value in EVA.}
\label{fig:EVA_p}
\centering
\end{figure}

\subsubsection{Varying Cache Size of a Small Cell}
In Fig. \ref{fig:EVA_K} the total rewards are depicted as the cache size of a small cell varies from 1 to 10. 
We set the number of enhanced views to 10 (i.e., $|\mathcal{E}| = 10$), and the number of required enhanced views for every user is 2. 
Note that since the number of desired enhanced views for every user is 2, the maximum total reward in this case is 100. 
When the cache size of a small cell increases, the total rewards of all schemes increase as expected.

\begin{figure}
\centering
  \begin{subfigure}[b]{0.5\columnwidth}
        \centering
        \includegraphics[scale=0.24]{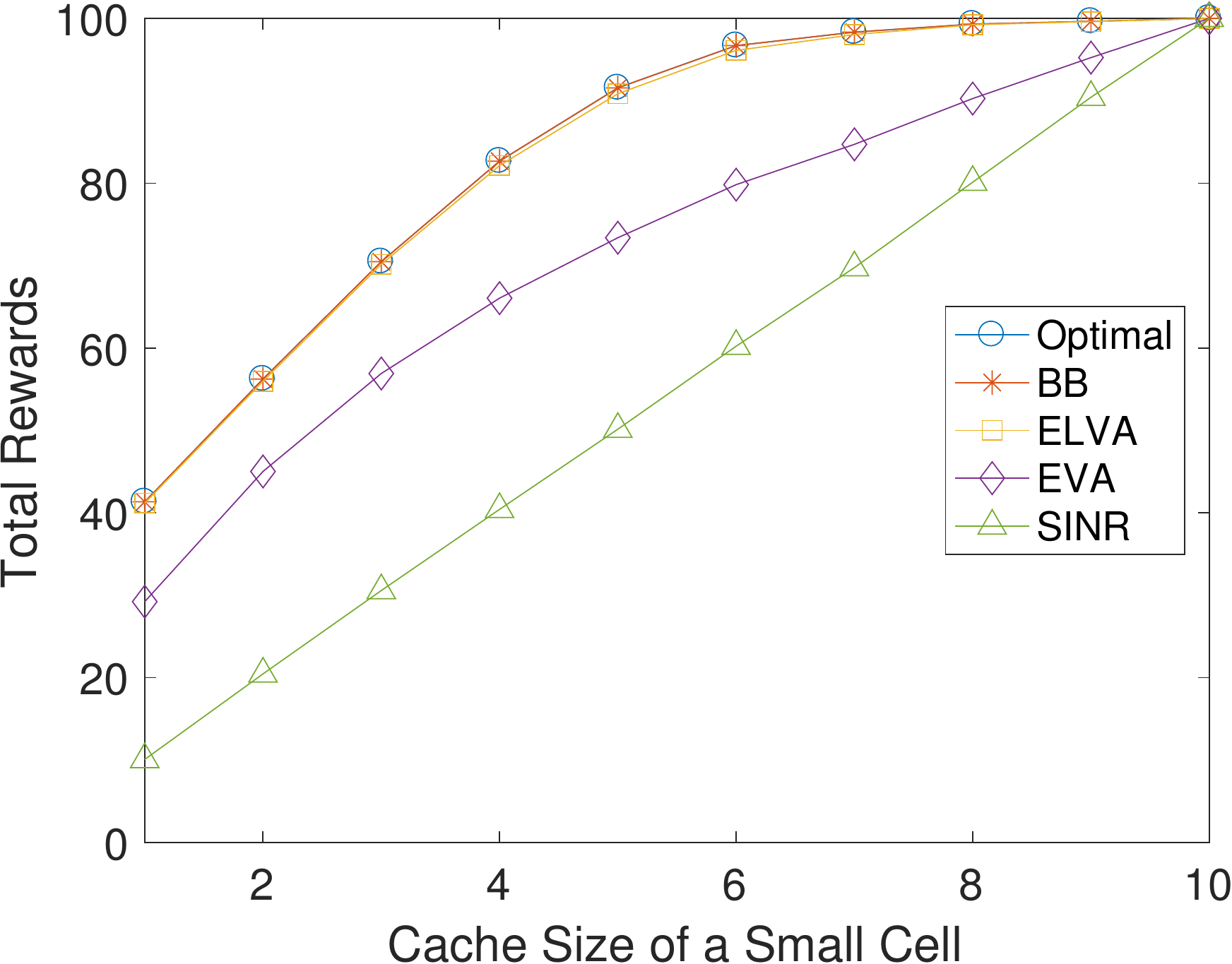}
        \caption{Hotspot/Uniform}
    \end{subfigure}%
    ~ 
    \begin{subfigure}[b]{0.5\columnwidth}
        \centering
        \includegraphics[scale=0.24]{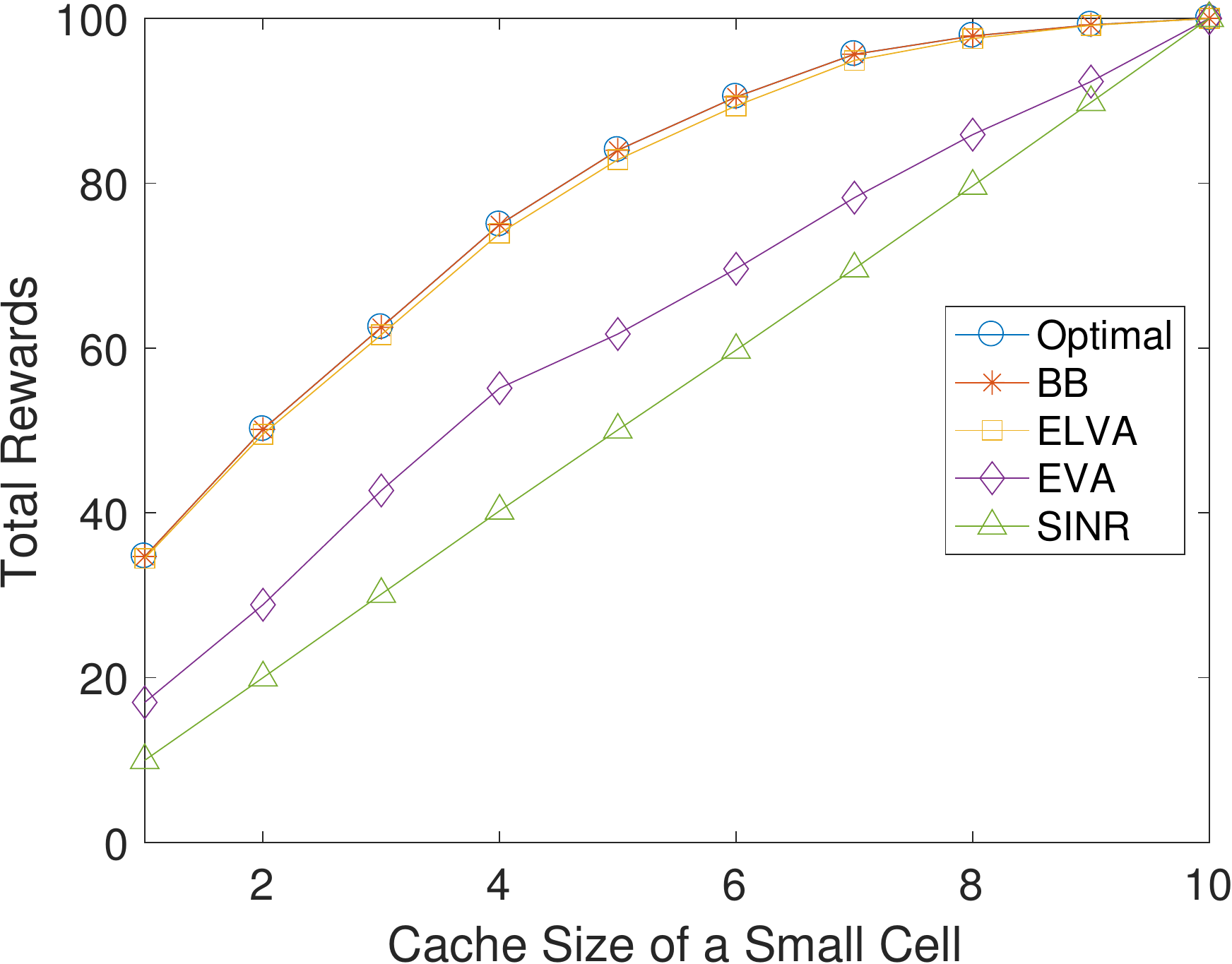}
        \caption{Uniform/Uniform}
    \end{subfigure}
\caption{Total rewards versus cache size of a small cell.}
\label{fig:EVA_K}
\centering
\end{figure}

\subsection{Simulation Results - Large-Scale Study}
\begin{figure}
\centering
  \begin{subfigure}[b]{0.5\columnwidth}
        \centering
        \includegraphics[scale=0.24]{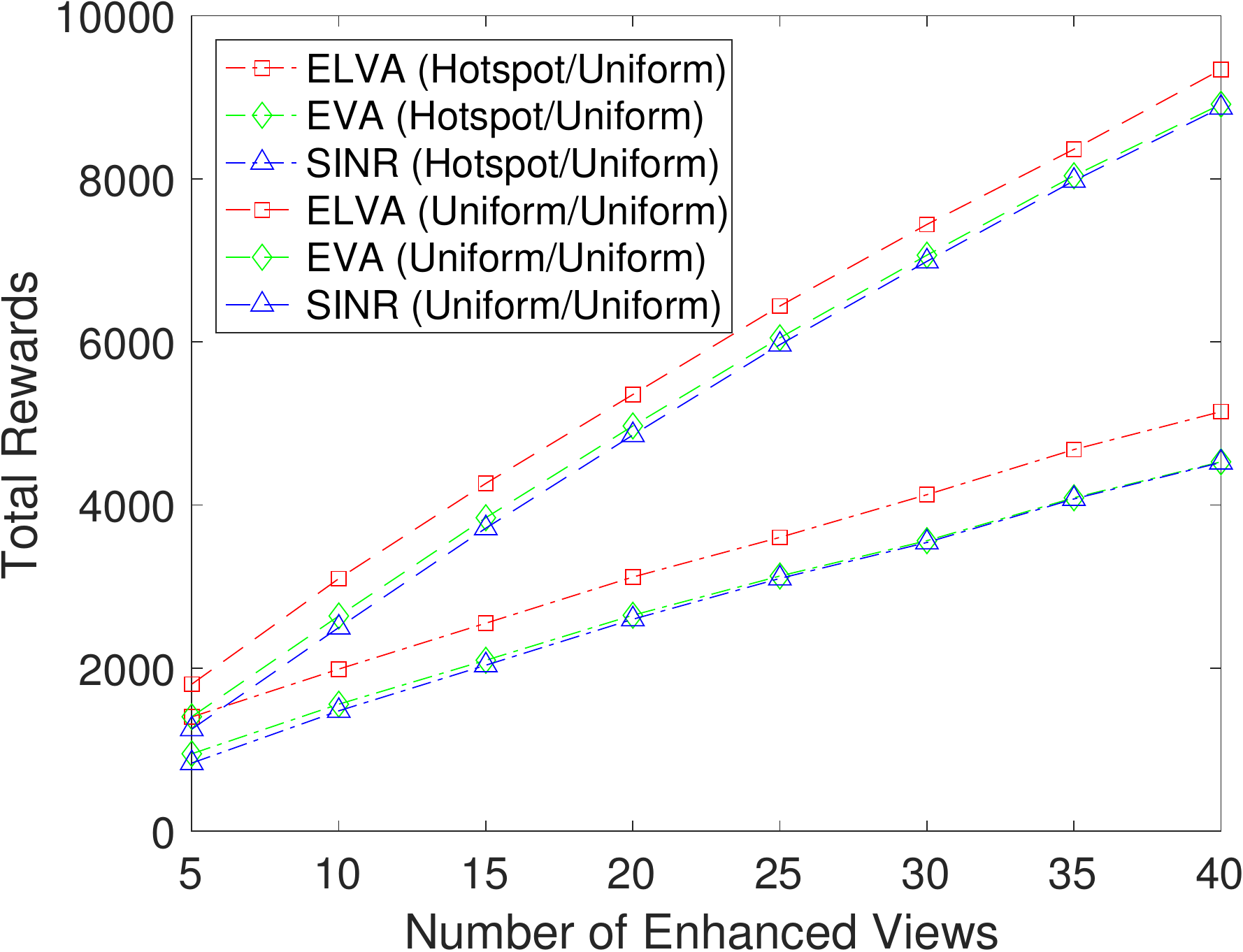}
        \caption{Total rewards versus number of enhanced views.}
    \end{subfigure}%
    ~
    \begin{subfigure}[b]{0.5\columnwidth}
        \centering
        \includegraphics[scale=0.24]{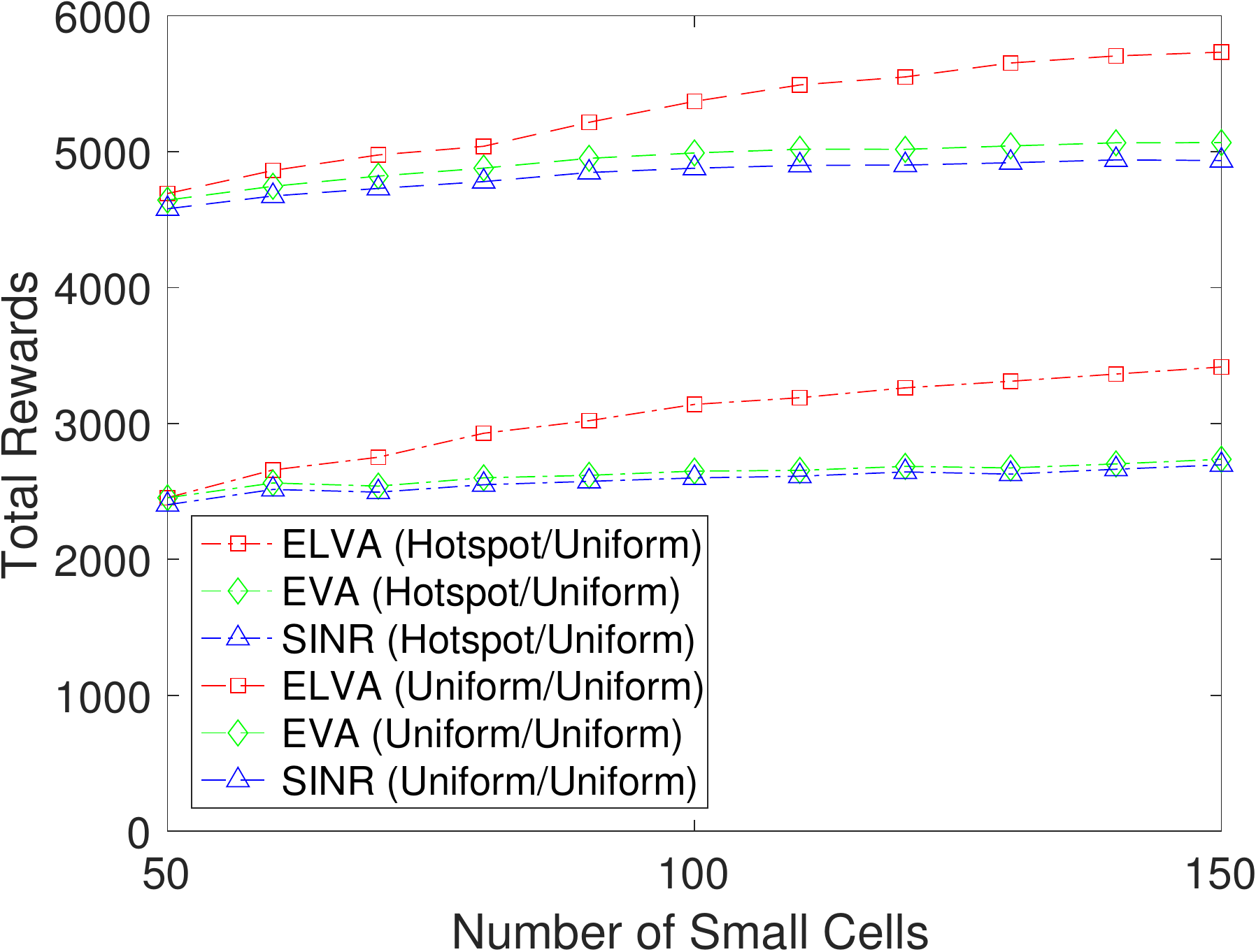}
        \caption{Total rewards versus number of small cells.}
    \end{subfigure}
    
     \begin{subfigure}[b]{0.48\columnwidth}
        \centering
        \includegraphics[scale=0.24]{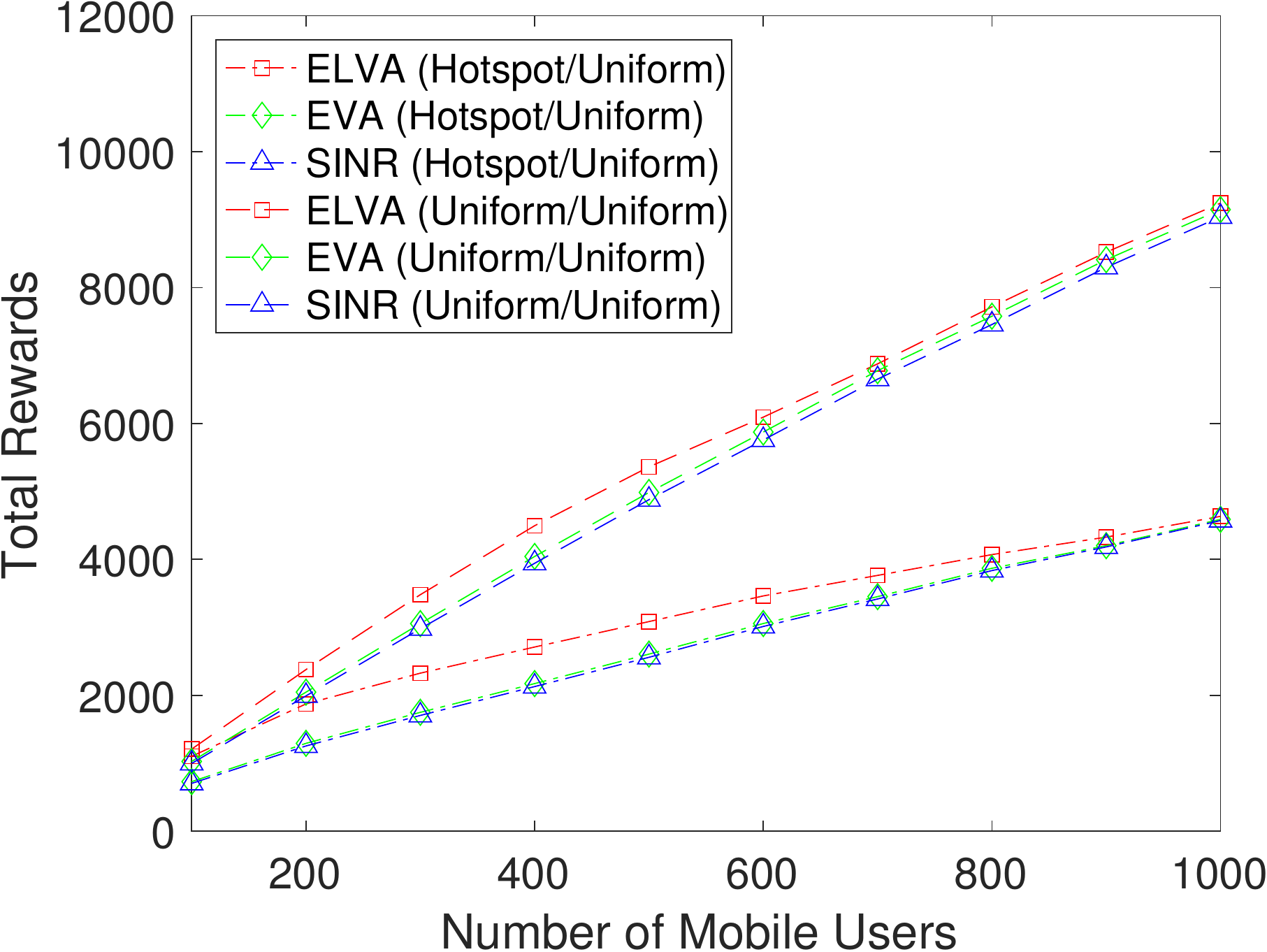}
        \caption{Total rewards versus number of mobile users.}
    \end{subfigure}
    ~
     \begin{subfigure}[b]{0.48\columnwidth}
        \centering
        \includegraphics[scale=0.25]{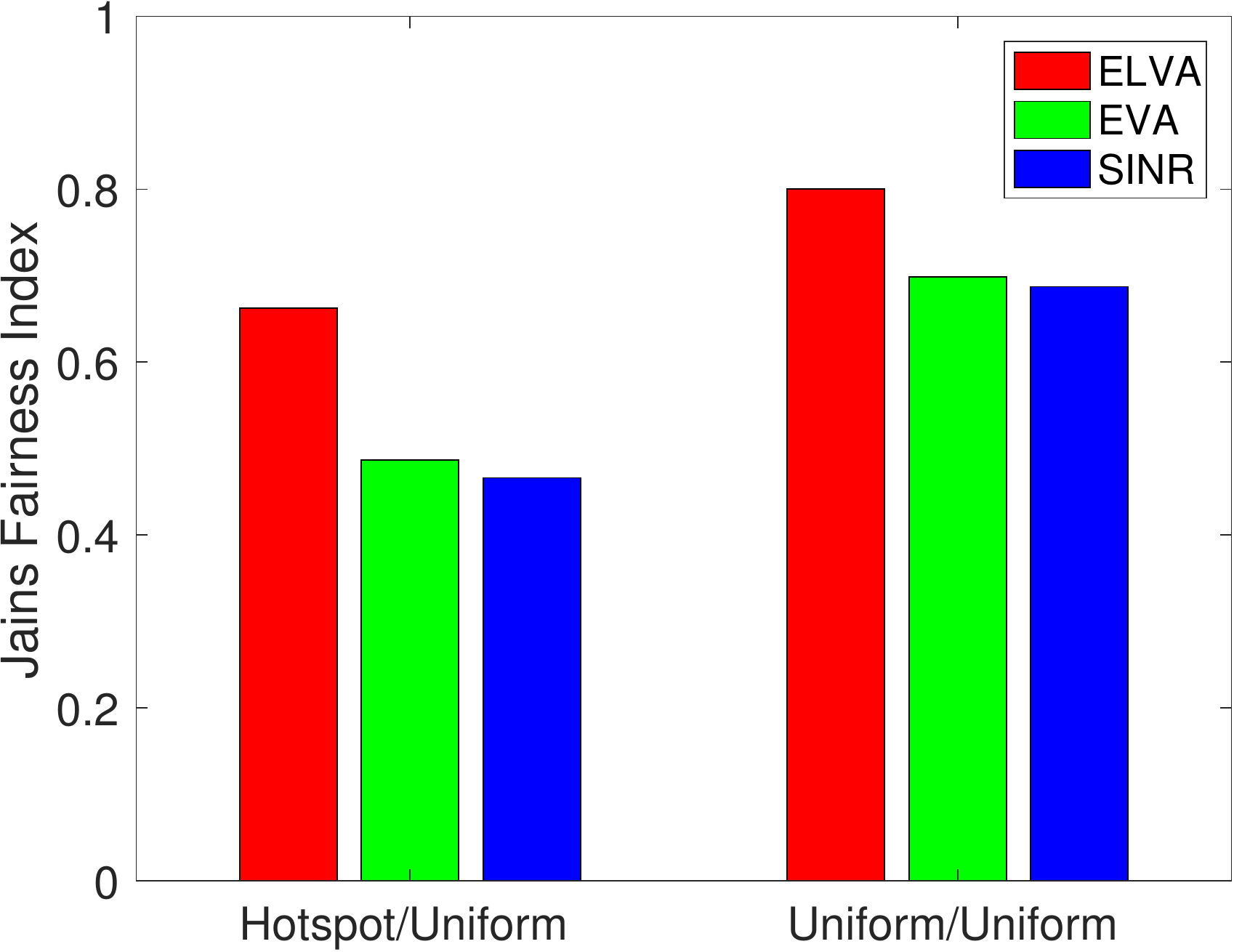}
        \caption{Jain's fairness index in the default settings.}
    \end{subfigure}
\caption{Large-scale simulation.}
\label{fig:EVA_Large}
\centering
\end{figure}
For the large-scale setup
the default number of small cells is 100 (which corresponds to about 35 small cells per square km, a dense deployment which is nevertheless less than the 75-200 small cells per km envisioned by
industry \cite{5GAmerica,Nokia}), the default number of enhanced views is 20, and the default number of VR users is 500. 
The same scenarios are considered as those in the small-scale network: ``Hotspot/Uniform" and ``Uniform/Uniform".

We vary the number of enhanced views from 5 to 40 in Fig. \ref{fig:EVA_Large} (a), the number of small cells from 50 to 150 in Fig. \ref{fig:EVA_Large} (b), 
and the number of mobile users from 100 to 1000 in Fig. \ref{fig:EVA_Large} (c). 
As shown in the figures, the trend of the total rewards under different varying parameters in the large-scale study is the same as those in the small-scale study.
This implies that the study in the small-scale study is directly applicable to that in the large-scale study.
In Fig. \ref{fig:EVA_Large} (d), we evaluate the fairness of the users using Jain's index \cite{JainIndex} in the default settings. As expected, ``ELVA" outperforms ``EVA" and ``SINR" because some users in ``EVA"  and ``SINR" are more likely to receive small rewards.

In summary, it is evident that SINR-based user association is not the best strategy for two-tier 360-degree video delivery, and the proposed polynomial time ``ELVA" 
algorithm achieves a near-optimal performance in all the scenarios.

\section{Extension} 
\label{Extension}
In this section, we briefly discuss some mathematical extensions of the problem of optimal user association and resource allocation for two-tier 360 video delivery (Problem (\ref{p1})). 
\subsection{Discrete-Level Video Coding} \label{DiscreteCoding}
Assume there are $F$ levels of video quality \footnote{This kind of video coding is applied in H.264. For example, H.264 uses four levels to represent a video \cite{H264}.}. We can reformulate Problem (\ref{p1}) as follows:
\begin{align}
\max_{x_{i,j},y_{i,j,k}}&~~\sum_{\forall i \in \mathcal{M}}{\sum_{\forall j \in \mathcal{S}}{\sum_{\forall k \in \mathcal{E}}{\frac{y_{i,j,k}}{F}w_{i,j,k}}}} \nonumber\\
\text{subject to:} &~~\sum_{\forall j \in \mathcal{S}}{x_{i,j}} = 1, \;\; \forall i \in \mathcal{M}, \nonumber\\
&~~\frac{y_{i,j,k}}{F} \leq x_{i,j}w_{i,j,k},  \;\; \forall i \in \mathcal{M}, \forall j \in \mathcal{S}, \forall k \in \mathcal{E}, \nonumber\\
&~~\max_{\forall i \in \mathcal{M}}\left\{x_{i,j}N^b_{i,j}\right\}+\sum_{\forall i \in \mathcal{M}}\sum_{\forall k \in \mathcal{E}}{\frac{y_{i,j,k}}{F}N^e_{i,j,k}} \leq N_j, \nonumber\\
&~~~~~~~~~~~~~~~~~~~~~~~~~~~~~~~~~~~~~~~~~~~~~~~~~~\forall j \in \mathcal{S}, \nonumber\\
&~~x_{i,j} = \{0,1\}, \forall i \in \mathcal{M}, \forall j \in \mathcal{S}, \nonumber\\
&~~y_{i,j,k} \in \{0,1,....,F\}, \forall i \in \mathcal{M}, \forall j \in \mathcal{S}, \forall k \in \mathcal{E}. \label{discrete_p}
\end{align}
Note that this problem cannot be solved by the proposed algorithms directly. Instead, we could relax this problem by making $y_{i,j,k} \in [0, F]$ and solve the relaxed problem by the proposed algorithms. We then round the value of $y_{i,j,k}$ to the closest level $0, 1, ..., F$ and access the outcome.

\subsection{Utility-Based Optimization} \label{UtilityOptimization}
Suppose $\mathcal{U}(\cdot)$ is the utility function mapping the received video quality to the quality of experience of a user. Then, we can reformulate Problem (\ref{p1}) as follows:
\begin{align}
\max_{x_{i,j},y_{i,j,k}}&~~\sum_{\forall i \in \mathcal{M}}{\sum_{\forall j \in \mathcal{S}}{\sum_{\forall k \in \mathcal{E}}{\mathcal{U}(y_{i,j,k}w_{i,j,k})}}} \nonumber\\
\text{subject to:} &~~\sum_{\forall j \in \mathcal{S}}{x_{i,j}} = 1, \;\; \forall i \in \mathcal{M}, \nonumber\\
&~~y_{i,j,k} \leq x_{i,j}w_{i,j,k},  \;\; \forall i \in \mathcal{M}, \forall j \in \mathcal{S}, \forall k \in \mathcal{E}, \nonumber\\
&~~\max_{\forall i \in \mathcal{M}}\left\{x_{i,j}N^b_{i,j}\right\}+\sum_{\forall i \in \mathcal{M}}\sum_{\forall k \in \mathcal{E}}{y_{i,j,k}N^e_{i,j,k}} \leq N_j, \nonumber\\
&~~~~~~~~~~~~~~~~~~~~~~~~~~~~~~~~~~~~~~~~~~~~~~~~~~\forall j \in \mathcal{S}, \nonumber\\
&~~x_{i,j} = \{0,1\}, \forall i \in \mathcal{M}, \forall j \in \mathcal{S}, \nonumber\\
&~~y_{i,j,k} \in [0,1], \;\; \forall i \in \mathcal{M}, \forall j \in \mathcal{S}, \forall k \in \mathcal{E}. \label{utility_p}
\end{align}
Note that if the utility function, $\mathcal{U}(\cdot)$, is convex, we can reapply all the proposed algorithms to solve this problem. 
\section{Conclusion} 
\label{Conclusion}
We jointly optimized user-cell association and resource allocation for delivering two-tier 360 video in wireless virtual/augmented reality. 
We formulated the problem using mixed integer linear programming, proved it is a NP-hard, described an optimal algorithm and proposed
a polynomial time approximation algorithm which was shown to be near optimal in practice.
Simulation results also established that the proposed algorithm can boost user experience by at least 30\% compared to baseline user association schemes. 
\bibliographystyle{ieeetr}
\bibliography{ref/reference}
\end{document}